%% file: main.tex
\documentclass[conference,compsoc]{IEEEtran}
\IEEEoverridecommandlockouts
\usepackage{myresearchpack}
\usepackage{tikz}
\usepackage{amsmath}
\usepackage{myresearchpack}
\usepackage{multicol}
\usepackage{balance}
\usepackage{tikz} 
\usetikzlibrary{positioning}
\usepackage{booktabs}
\usepackage{graphicx}
\usepackage{subcaption}
\usepackage{hyperref}
\usepackage{multirow}

\ifCLASSOPTIONcompsoc
  \usepackage[nocompress]{cite}
\else
  \usepackage{cite}
\fi

\newcommand{\negl}{\mathsf{negl}}
\newcommand{\PPT}{\textsf{PPT}}

\newcommand{\PCS}{\mathsf{PCS}}
\newcommand{\Prove}{\mathsf{Prove}}

\newcommand{\Rzkml}{\mathcal{R}_{\mathsf{zkml}}}

\newcommand{\zkmod}{\mathsf{\X}}

\begin{document}

\author{
\IEEEauthorblockN{
Pawan Kumar Sanjaya\IEEEauthorrefmark{1},
Christina Giannoula\IEEEauthorrefmark{2}\IEEEauthorrefmark{4}\thanks{\IEEEauthorrefmark{4}Work partially conducted while the author was at the University of Toronto.},
Valdy Oktavian\IEEEauthorrefmark{1},
Mehdi Saeedi\IEEEauthorrefmark{3},
Gabor Sines\IEEEauthorrefmark{3},\\
Gururaj Saileshwar\IEEEauthorrefmark{1},
Nandita Vijaykumar\IEEEauthorrefmark{1}
}
\IEEEauthorblockA{\IEEEauthorrefmark{1}University of Toronto, Canada}
\IEEEauthorblockA{\IEEEauthorrefmark{2}Max Planck Institute for Software Systems (MPI-SWS), Germany}
\IEEEauthorblockA{\IEEEauthorrefmark{3}Advanced Micro Devices (AMD)}
\thanks{\textcopyright~2026 Advanced Micro Devices, Inc. All rights reserved. AMD, the AMD Arrow logo, Radeon, and combinations thereof are trademarks of Advanced Micro Devices, Inc. Other product names used in this publication are for identification purposes only and may be trademarks of their respective companies.}
}

\shortname{zkComposer}
\title{\X: Decomposing Proof Construction to Scale zkML}

\maketitle

\begin{abstract}
Zero-knowledge machine learning (zkML) enables a server to perform verifiable inference while keeping model parameters private from the client. However, existing zkML systems incur prohibitive proof-generation costs. We observe that proof generation exhibits limited parallelism; that is, prover time does not decrease significantly as the number of threads increases. This limitation is because existing systems rely on monolithic proof computation, constructing a single proof for the entire machine learning model.

We introduce \X, a modular proof-construction framework that unlocks an additional dimension of parallelism, in addition to the parallelism in existing proof kernels. 
\X decomposes the zkML proof of correct inference into independent sub-proofs, each covering a subset of the computation for inference, \eg each independent sub-proof can cover a subset of contiguous layers in the ML model. 
Adjacent sub-proofs are cryptographically linked through shared commitments to the activations from the boundary layer. 
\X provides the same guarantees as the monolithic proof without requiring additional linking proofs or changes to the underlying cryptographic primitives.

We implement \X and evaluate it on three CNNs and GPT-2. We show that, on CNN workloads, \X reduces prover time and response time by up to $3.25\times$ relative to zkCNN~\cite{zkcnn}. On GPT-2, \X reduces these times by up to $4.83\times$ relative to zkGPT~\cite{zkgpt}, when partitioning along the model layers. When partitioning across both model layers and input sequences in GPT-2, we show that \X reduces prover time and response time by up to $6.84\times$ relative to zkGPT~\cite{zkgpt}.



\end{abstract}


\input{src/introduction}
\input{src/background}

\input{src/key_idea}

\input{src/prot_desc}
\input{src/security_analysis}

\input{src/discussion}
\input{src/evaluation}

\input{src/prior}

\input{src/conclusion}

\bibliographystyle{IEEEtran}
\bibliography{ref}
\appendices

\input{src/appendix}

\end{document}

%% file: src/introduction.tex
\section{Introduction}
\label{sec:introduction}


Machine learning (ML) inference systems are increasingly deployed in high-stakes domains including personalized recommendations~\cite{zhang2019deep}, fraud detection~\cite{abdallah2016fraud}, insurance claim processing~\cite{kuo2019deep}, and healthcare diagnostics~\cite{esteva2019guide}. In such settings, it is often critical that a response is generated by a specific model. Such a model may have been verified to satisfy formal properties such as fairness~\cite{MATOS2026101001}, or certified to meet a required quality of output. These guarantees do not hold if a different model is substituted. However, it is difficult to ensure that the service providers who host and deploy these models actually use the specified model~\cite{cai2025gettingpayforauditing}. For example, a provider could substitute a cheaper model in place of the one specified by the user. Revealing the model weights would let users rerun the inference and validate the response. However, weights are trade secrets, requiring significant financial resources to train (\eg $\$40$M for GPT-4~\cite{gpt_cost}) and cannot be revealed~\cite{zkcnn,zkgpt,zkml}. Thus there is a need for a mechanism that lets users verify that a response was generated by the specified model, without knowing the model weights.

Recent advances in cryptography, particularly zero-knowledge proofs (ZKPs)~\cite{11023367,libra,lasso}, offer a practical mechanism to enforce these guarantees.
A ZKP allows a \emph{prover} to convince a \emph{verifier} that a given computation was executed correctly \emph{without revealing} any private inputs, intermediate values, or other sensitive information beyond what is implied by the output alone~\cite{goldwasser1989knowledge}. When used to verify the model used in ML, such systems are commonly referred to as \emph{zkML}~\cite{lee2024vcnn,zkcnn,zkml,qu_verfcnn_2025,ghodsi2017safetynets,cryptoeprint:2024/162}. Figure~\ref{fig:setting} illustrates a typical zkML inference protocol in the server-client setting~\cite{zkcnn,zkml,zkgpt,zkllm,apollo}. The client submits a query $q$ to the server.
The server runs inference with the model specified by the client using its secret model parameters $\theta$ to obtain the response $r$.
It then returns $r$ and a proof $\pi$.
The client verifies $\pi$ to confirm that $r$ is the correct output of the specified model on the input $q$, without learning anything about the model parameters $\theta$. As in prior works~\cite{zkcnn,zkgpt,qu_verfcnn_2025,zkllm,zkml}, we assume that the model architecture itself is public, \ie known by the client.

\begin{figure}
    \centering
    \includegraphics[trim={15mm 8mm 20mm 2mm}]{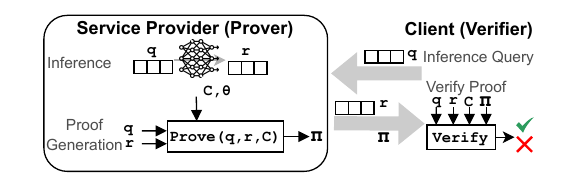}    
    \vspace{0.5cm}
    \caption{Overview of a zkML system: service provider generates a proof for inference.}
    \label{fig:setting}
\end{figure}

To generate the proof, the server and the client first encode the ML operations (\eg matrix multiplications, activations) as arithmetic constraints called the circuit $C$.
The inference server demonstrates that the response $r$ satisfies $C$, implying correct execution of the specified model on $q$. Here, the server is the prover and the client is the verifier. The model parameters $\theta$ are secret and known only to the server, \ie they must not be inferred by the client. Thus, the prover generates a \emph{commitment} to the model parameters and intermediate activations using a cryptographic commitment scheme. The commitment scheme enables the verifier to check that the underlying values satisfy $C$ while keeping those values hidden. In this way, zkML supports ML services where trust is grounded in cryptographic guarantees rather than blind faith in the service provider. 

Computing this cryptographic proof is substantially more expensive than ordinary inference. For example, our measurements show that zkGPT~\cite{zkgpt} requires $147.7$s to generate a proof for 64 tokens on a 192-core machine with 755 GB of RAM, while generating the corresponding model response on the same machine takes only 94.55ms. This proof-generation overhead severely degrades user experience in any deployed zkML system. Recent systems reduce this cost via ML- and cryptography-aware optimizations. Quantization~\cite{feng2021zen}, lookup-table approximations of non-linear functions~\cite{zkllm,zkgpt}, and efficient cryptographic representations of common ML kernels~\cite{feng2024zeno,zkcnn,zkgpt} lower the proving cost of each arithmetic operation. Leading frameworks such as zkGPT~\cite{zkgpt} also use \emph{circuit squashing}, which combines the computations of many model layers into a single circuit layer and reduces the overall circuit depth. Despite these advances, real-world zkML deployments still incur substantial proof-generation times.

\begin{figure}
    \centering
    \begin{adjustbox}{max totalsize={\linewidth-50pt}{\textheight}}
    \includegraphics[trim={25mm 5mm 20mm 5mm},]{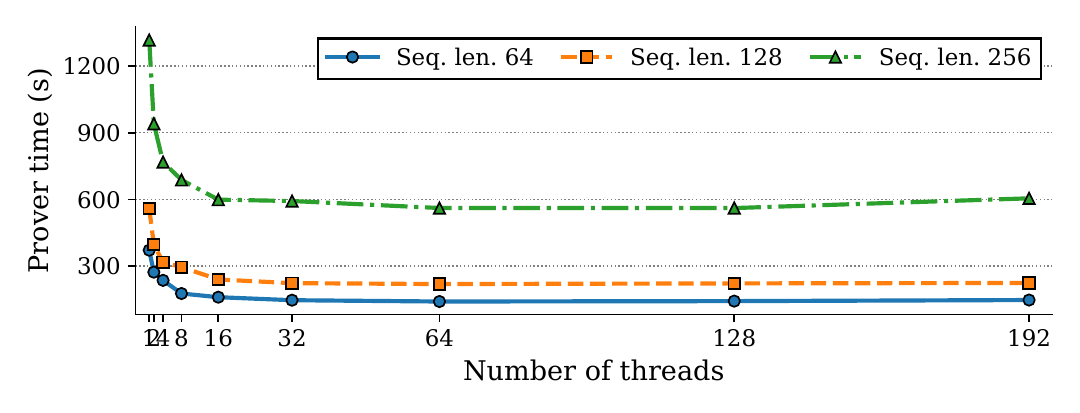}    
    \end{adjustbox}
    \caption{Prover time of zkGPT~\cite{zkgpt} does not decrease proportionally as the thread count increases.}
    \label{fig:zkgpt_scaling}
\end{figure}

Existing zkML systems~\cite{zkcnn,zkgpt,zkllm,qu_verfcnn_2025,zkml,ezkl} generate a single monolithic proof using an optimized arithmetic circuit for the entire model. This design exposes only intra-proof parallelism: multiple threads can execute the sumcheck protocol or compute commitments concurrently, but this parallelism saturates quickly. \Cref{fig:zkgpt_scaling} shows the prover time of zkGPT~\cite{zkgpt} when using different thread counts on a machine with 192 cores for several sequence lengths. Prover time stops decreasing beyond 16 threads because the workload is memory-bound, and additional threads stall on accesses to a large working set in main memory. The total memory footprint of the prover also grows with the size of the inference computation. For example, increasing the sequence length in zkGPT from $64$ to $256$ raises peak memory consumption of the prover from $240.5$ GB to $532.4$ GB. This restricts the models and input sizes for which proofs can be computed on memory-constrained accelerators and GPUs.

\noindent
\textbf{Goal.}
Our goal is twofold: 1) reduce zkML proof-generation latency, and 2) reduce the peak memory required for proof generation when using large input sizes, so that zkML can run on memory-constrained accelerators and GPUs. We propose \X, a new construction that, to our knowledge, is the first to expose parallelism \emph{across} independent sub-proofs in zkML frameworks. This reduces proof-generation time by complementing the parallelism exploited by prior work~\cite{zkgpt,zkcnn}.

In \X, we divide the inference computation along two dimensions: model layers and the input sequence. We refer to each resulting sub-computation as a \emph{partition}. When partitioning along the model layers, each partition contains the operations of a contiguous subset of layers. When partitioning along the input sequence, each partition contains the operations that compute the activations for a contiguous subset of input positions. Then, instead of generating a single proof for the entire model, we generate a separate proof for each partition. Thus, the circuit for each partition represents a subset of the operations required to compute the response. For example, when partitioning along the model layers, each sub-proof proves that the input was transformed using a subset of contiguous model layers. Within each partition, \ie within each sub-proof, \X retains an execution structure similar to that of prior works~\cite{zkgpt,zkcnn} and can still benefit from existing circuit-level optimizations such as circuit squashing. Across partitions, however, \X exposes sub-proof-level parallelism that is not available in a monolithic proof: different partitions can be proved concurrently. This can reduce end-to-end proof-generation time if all the sub-proofs together provide the same security guarantees as the monolithic proof. Thus, \X reduces proof-generation latency by proving partitions in parallel. Additionally, because each partition represents only a subset of the computation relative to the monolithic proof, each sub-proof requires less memory. The smaller memory footprint of each sub-proof also enables proof generation on memory-constrained accelerators and GPUs when each sub-proof is executed in sequence.

\noindent
\textbf{Challenges.}
Splitting proofs into multiple sub-proofs while providing the same security guarantees as existing zkML frameworks introduces two key challenges: (1) The integrity of the inference computation must be preserved across partition boundaries, \ie the proof must demonstrate that the output activations produced by partition $i$ (ordered from output to input) are \emph{equal} to the input activations in partition $i-1$; (2) The zero-knowledge property must be preserved, \ie the verifier must learn nothing about intermediate activations at partition boundaries beyond what is implied by the final output and model architecture.

To ensure the correctness of the computation across partition boundaries, we generate a single cryptographic commitment for the output activations of each partition and reuse this commitment in both sub-proofs of adjacent partitions. Since a commitment cryptographically binds to a unique value, this reuse creates a guarantee: the boundary activations used in both sub-proofs are identical. Specifically, the output activations produced in partition $i$'s proof are provably equal to the input activations used in partition $i-1$'s proof. To preserve zero knowledge, the commitments are masked with fresh random values. Here, the number of openings for each commitment across all sub-proofs is determined by the model architecture. Since the model architecture is public and known to both the prover and the verifier, the number of openings for each commitment is known at deployment. Thus, the prover uses this information to select a sufficient number of variables in the masking polynomial to mask all openings. This ensures that sharing commitments between sub-proofs reveals no additional information about model weights.

Our implementation of \X organizes proof generation into two phases: a sequential \emph{commitment phase} followed by a fully parallel \emph{proof phase}. In the commitment phase, the prover generates (masked) commitments to the boundary values of each partition. Once all boundary commitments are computed, the sub-proofs become entirely independent. As a result, the prover can generate all sub-proofs concurrently on separate processing cores.

\X can be implemented using the protocols of modern GKR-based zkML frameworks~\cite{qu_verfcnn_2025,zkcnn,zkgpt,zkllm}, and can be trivially extended to others~\cite{ezkl,zkml}. We implement \X using the arithmetization and polynomial commitment schemes from prior works, zkCNN~\cite{zkcnn} and zkGPT~\cite{zkgpt}. We compare our implementation with the respective baselines (zkCNN and zkGPT) by running them with the same hardware resources. For three CNNs, \X reduces prover time by up to $3.25\times$ relative to zkCNN with negligible impact on verifier time. Similarly, for GPT-2, \X achieves up to a $6.84\times$ prover-time speedup relative to zkGPT. \X also scales better with the number of available cores: its speedup over zkGPT increases from $2.87\times$ at 16 cores to $4.83\times$ at 192 cores with 12 partitions. Beyond exposing parallelism, this decomposition also reduces the prover's peak memory: generating sub-proofs sequentially keeps only one partition in memory at a time, yielding up to $8.1\times$ lower peak memory on GPT-2 (\S\ref{app:seqmem}).

\noindent\textbf{Contributions.} We make the following contributions:
\begin{enumerate}
\item We introduce \X, a modular proof construction for zkML inference that replaces the traditional monolithic proof with multiple independent sub-proofs. By partitioning the inference computation, \X exposes proof-level parallelism that is unavailable in prior zkML frameworks, while remaining compatible with existing circuit-level optimizations.
\item We design \X such that it provides the same security guarantees as a monolithic proof. \X uses shared commitments to enforce equality of activations across partition boundaries. It masks these commitments to preserve zero knowledge. We prove that \X preserves the guarantees of the underlying zkML proof system without requiring any changes to its cryptographic primitives.

\item We implement \X on top of zkCNN~\cite{zkcnn} and zkGPT~\cite{zkgpt} and evaluate it on three CNN models and GPT-2~\cite{gpt}. \X reduces prover time by up to $3.25\times$ for CNNs and up to $6.84\times$ for GPT-2 when computing the proofs in parallel. \X also reduces peak prover memory by up to $8.1\times$ in sequential sub-proof generation. These results show that partitioned proof generation can improve zkML latency and enable proving on more memory-constrained hardware.
\end{enumerate}

%% file: src/background.tex
\section{Background}
\label{sec:background}


\subsection{Expressing ML Operators in ZkML}
\label{subsec:indexed-relations}
ZkML inference systems (Figure ~\ref{fig:setting}) enable a service provider to prove to a client that a response $r$ to query $q$ was computed using a specified ML model with weights $\theta$. The service provider acts as the \emph{prover} (\prover) and the client as the \emph{verifier} (\verifier). The inference computation is represented as a layered arithmetic circuit $C$ consisting of $d$ layers, indexed from $0$ (output layer) to $d$ (input layer). Each layer comprises gates that perform arithmetic operations on the output values from the layer $i+1$ or the input layer $d$. Typically, each ML layer is represented as one or more layers in $C$.

To generate the proof, \prover demonstrates knowledge of a private witness $w = \{\mathsf{adv}, \theta\}$ (advice values and model weights) such that the public instance $x = (q, r)$ satisfies the circuit constraints, \ie $C$ evaluates to $1$ on input $x$ and witness $w$. The advice values are used to compute non-linear activation operations such as ReLU~\cite{zkgpt} and max pooling~\cite{zkcnn}. Public values $q$ and $r$ are known to both parties, while only \prover knows the private witness. This proof convinces \verifier that $r$ is indeed the correct output obtained by evaluating the specified model on input $q$. In modern zkML systems~\cite{zkcnn,zkml,zkgpt,zkllm}, the values $(q,r,\mathsf{adv},\theta,\ldots)$ are represented over a prime field ($\mathbb{F}_p$), \ie integers $(0,\ldots,p-1)$ with addition and multiplication modulo $p$. These values are encoded into low-degree polynomials (represented as $f(.), g(.)$) over the field, so the circuit ($C$) can be viewed as a collection of polynomial constraints that the instance and witness polynomials must satisfy.

zkML systems must support two distinct classes of operations~\cite{zkgpt}. (1) \textbf{Linear kernels} (\eg matrix multiplication and convolution) are represented directly by low-degree polynomial relations linking their input and output polynomials. \prover proves that these relations hold. (2) \textbf{Non-linear activations} are typically handled by adding auxiliary witnesses (advice values) and lookup-style constraints: \prover proves that evaluations of the corresponding witness polynomials lie in a public table and that these values are consistent with the arithmetic constraints for the operation. For example, ReLU can be enforced using auxiliary boolean/range witnesses that certify the sign/range of the input and a simple arithmetic relation tying the output to the input.

\noindent
\textbf{Notation.} We use $f(\cdot)$, $h(\cdot)$ for polynomials, $x, y, z$ to denote vectors of variables, and $g, r, u, v$ for vectors of values. Subscript notation ($x_i$) denotes the $i$-th element of the vector. We use the standard notation for bitstrings $\{0, 1\}^{*}$, and fields $\mathbb{F}_p$ with order $p$. For a $d$-variate polynomial $f(x_1,\cdots,x_d)$, $\deg_i(f)$ is the degree of $f$ in $x_i$. A polynomial is \emph{multilinear} if $\deg_i(f) \leq 1$ for all $i \in [d]$. All adversaries are Probabilistic Polynomial-Time (PPT), unless stated otherwise.

\subsection{GKR Protocol}
\label{subsubsec:gkr}
The GKR protocol~\cite{GKR15, virgo} is used by \prover to show that the evaluation of an inference computation represented by a layered arithmetic circuit ($C$) is correct. Modern zkML systems~\cite{zkcnn,zkgpt,zkllm} use GKR to prove ML computations. Using existing convention~\cite{zkcnn,libra}, $S_i$ is the number of gates at layer $i$ and $s_i = \log_2 S_i$. Each gate at layer $i$ is uniquely identified by a Boolean string in the hypercube $\{0,1\}^{s_i}$. 
$\tilde{V}_i$ encodes all the outputs of layer $i$, \eg for a gate at index $z$, $\tilde{V}_i(z)$ is the output of the corresponding gate in layer $i$. We use wiring predicates for the connections between the outputs of a layer $i+1$ and the gates in layer $i$: Boolean selector functions evaluating to $1$ only if there is a wire from layer $i+1$ (or $d$) to the corresponding gate.   The wiring predicate between addition gates in layer $i$ and values from layer $i+1$ is $\widetilde{\mathrm{add}}_{i,i+1}$, \ie $\widetilde{\mathrm{add}}_{i,i+1}(z,b)=1$ only if the gate at index $z$ operates on the value at index $b$ of layer $i+1$. The other wiring predicates $\widetilde{\mathrm{add}}_{i,d}$,$\widetilde{\mathrm{mult}}_{i,d,d}$, $\widetilde{\mathrm{mult}}_{i,i+1,d}$, and $\widetilde{\mathrm{mult}}_{i,i+1}$ are defined analogously. Using these selectors, the value at any index $z$ for layer $i$ ($\tilde{V}_i(z)$) can be expressed as a single polynomial~\cite{zkcnn}:
{\small
\begin{align}
\tilde{V}_i(z) 
=&\sum_{b \in \{0,1\}^{s_{i+1}}} \widetilde{\mathrm{add}}_{i,i+1}(z,b) \cdot \tilde{V}_{i+1}(b) \nonumber \\ +&\sum_{b \in \{0,1\}^{s_{d}}} \widetilde{\mathrm{add}}_{i,d}(z,b) \cdot \tilde{V}_{i,d}(b) \nonumber\\
   +&\sum_{b,c \in \{0,1\}^{s_{i+1}}} \widetilde{\mathrm{mult}}_{i,i+1}(z,b,c) \cdot \tilde{V}_{i+1}(b) \cdot \tilde{V}_{i+1}(c) \nonumber\\
  +&\sum_{b,c \in \{0,1\}^{s_{d}}} \widetilde{\mathrm{mult}}_{i,d,d}(z,b,c) \cdot (\tilde{V}_{i,d}(b) \cdot \tilde{V}_{i,d}(c)) \nonumber\\
  +&\sum_{\substack{
b \in \{0,1\}^{s_{i+1}} \\
c \in \{0,1\}^{s_d}
}} \widetilde{\mathrm{mult}}_{i,i+1,d}(z,b,c) \cdot \tilde{V}_{i+1}(b) \cdot \tilde{V}_{i,d}(c),
  \label{eq:modified-gkr-wiring}
\end{align}}
where $\tilde{V}_{i,d}$ is the subset of values in the input layer $d$ that are used for computations of layer $i$. When evaluated on Boolean inputs, exactly one of the $\mathrm{add}$ or $\mathrm{mult}$ terms equals $1$ for each gate~$z$. 
This property guarantees that the sum precisely recovers the arithmetic value computed by every gate in layer~$i$. Equivalently, verifying the correctness of layer~$i$ reduces to checking that Eq.~\eqref{eq:modified-gkr-wiring} holds.  GKR reduces the problem of verifying an entire layer-$i$ computation to checking the validity of Eq.~\eqref{eq:modified-gkr-wiring} at a single random point~$z$ chosen by \verifier ~\cite{GKR15}.

At the beginning of the GKR protocol, \prover sends the claimed output value $\tilde{V}_0$ to \verifier. \verifier then selects a random field element $g^{(0)} \in \mathbb{F}_p$, evaluates $\tilde{V}_0(g^{(0)})$, and sends $g^{(0)}$ to the prover. Then, \prover and \verifier reduce the claim about $\tilde{V}_0$ to a claim about $\tilde{V}_1$ using Eq. ~\eqref{eq:modified-gkr-wiring}. This process continues layer by layer until \verifier obtains a claim about the input layer ($\tilde{V}_d$). Each reduction step is performed using the interactive sumcheck protocol recursively on Eq. ~\eqref{eq:modified-gkr-wiring}~\cite{GKR15}. The final claim on the input layer is verified by giving \verifier oracle access to $\tilde{V}_d$. \verifier can query an oracle, a black-box function (see \S\ref{subsec:pcs}) that knows the polynomial $\tilde{V}_d$. If the value from the oracle query matches the final claim,  \verifier is convinced that the computation was performed correctly.

\subsection{Interactive Sumcheck Protocol}
\label{subsec:sumcheck}
The sumcheck protocol~\cite{sumcheck} allows \prover to show that the sum of the evaluations of a polynomial $h$ (with $m$ variables) over binary inputs is $H$. Formally, it proves~\cite{sumcheck}
\begin{align}
\label{eq:sumcheck}
H \stackrel{?}{=} \sum_{b\in\{0,1\}^m} h(b).    
\end{align}

It is a public-coin protocol where the random values sampled by \verifier in each round are sent to \prover~\cite{sumcheck}. It proceeds in $m$ rounds where in each round \prover proves that the claim made in the previous round is true using the random challenge values. In the first round, \prover sends a univariate polynomial~\cite{sumcheck}
\[
q_1(z) \;=\; \sum_{b_{2},\dots,b_m \in \{0,1\}} h(z, b_2,\dots,b_m).
\]
At the end of the round, \verifier checks that $H=q_1(0)+q_1(1)$ and sends a random value $r_1 \in \mathbb{F}$ to \prover. The claim for the next round is set to $q_1(r_1)$ and the same steps are repeated for $m$ rounds. Finally, \verifier has a claim $q_m(r_m)$. This claim is verified by performing an oracle query for the evaluation of $h$ at the point $(r_1, \dots, r_m)$. \verifier checks that this evaluation equals the claim ($q_m(r_m)$). If the equality holds,  \verifier is convinced of the original claim (Eq.~\eqref{eq:sumcheck}). 

When used in the GKR protocol, at the end of the sumcheck protocol on $\tilde{V}_0$, \verifier queries \prover for the values of $\tilde{V}_1$ at two random points $u$ and $v$~\cite{GKR15}. \verifier evaluates the wiring predicates and computes $\tilde{V}_0$ according to Eq.~\eqref{eq:modified-gkr-wiring}. It then verifies that this computed value matches the final claim in the sumcheck protocol. This verification step reduces the original claim for $\tilde{V}_0$ to two separate claims for $\tilde{V}_1$, which are combined into a single claim via random linear combination~\cite{zksumcheck}. The procedure repeats iteratively, eventually yielding a claim for $\tilde{V}_d$.

\subsubsection{Polynomial Commitment Scheme}
\label{subsec:pcs}
A Polynomial Commitment Scheme (PCS)~\cite{hyrax,spartan} emulates the oracle access in GKR  using 4 algorithms ($\mathsf{Setup}$, $\mathsf{Commit}$, $\mathsf{Open}$, and $\mathsf{Verify}$). It allows  \prover to commit (using $\mathsf{Commit}$) to a polynomial $f$, and later can only produce evaluations  of $f$ using that commitment (using $\mathsf{Open}$), \ie it cannot change the polynomial after commitment. \verifier learns nothing about the polynomial (\eg coefficients of the monomials) apart from the evaluation itself. Thus, by integrating a PCS as the oracle, \prover can convince  \verifier that the computation was correct. We formally define the algorithms in the PCS in \S\ref{sec:pcs_formal}. 

A PCS must satisfy (1) \emph{completeness}: honestly generated proofs verify; (2) \emph{binding}: a commitment $\mathsf{com}$ uniquely determines a degree-bounded polynomial $f$, such that every accepting opening $(\alpha,v,\pi)$ satisfies $v=f(\alpha)$; and (3) \emph{hiding}: the commitment $\mathsf{com}$ does not reveal $f$ beyond what is implied by opened evaluations.

\subsubsection{Zero-Knowledge GKR}
The standard GKR protocol provides soundness for the claimed computation, but it is not zero-knowledge~\cite{libra}. If it is executed naively, \verifier can learn information about the model weights $\theta$. There are two sources of information leak when executing the GKR protocol using sumcheck style arguments for each layer. First, it is from the univariate polynomials sent during each round of the sumcheck protocol. \verifier can learn the evaluation of each of those polynomials at a random point. These values represent linear combinations of the layer values $\tilde{V}_i$, which allows \verifier to partially reconstruct some entries of the layer. This is addressed by adding a random polynomial to Eq. ~\eqref{eq:modified-gkr-wiring} (see \S\ref{sec:zk_sumcheck}). Second, at the end of the sumcheck for each layer, \verifier learns the evaluation of $\tilde{V}_i$ at two points. Each point corresponds to a layer output. Thus, the polynomials for each layer are replaced with their low degree extensions (LDE) as follows~\cite{libra}: 
\begin{align}
    \dot{V}_{i}(z) = & \tilde{V}_{i}(z)+Z_{i}(z)\sum_{w \in \{0,1\}}R_{i}(z_1,w),
    \label{eq:lde}
\end{align}
\vspace{-0.75cm}
\begin{align}
    \textrm{ where } 
    &Z_i(z) =  \prod_{i=1}^{s_i}z_i(1-z_i) \textrm{ and } \nonumber \\
    R_i(z_1,w) = &a_0 + a_1z_1+a_2w+a_3z_1w+a_4z_1^2+ \nonumber\\ &a_5w^2+a_6z_1^2w^2 .
    \label{eq:masking-2}
\end{align}
$R_i(z_1,w)$ has a degree 2 for $a_0,\dots,a_6$ chosen randomly by \prover. Since $Z_i(z)=0$ for all $z\in\{0,1\}^{s_i}$, $\dot{V_i}$ agrees with $\tilde{V_i}$ over the boolean hypercube and can be substituted in Eq. ~\ref{eq:modified-gkr-wiring}. Each evaluation for $\tilde{V}_{i}$ is now masked by an evaluation of $R_i$. Thus, \verifier learns nothing about the individual values of layer polynomials from the evaluations.

\subsection{GKR-based Zk-SNARKs}
\label{subsec:zksnark}

Formally zkML proving can be represented via an indexed relation $\mathcal{R}$, where the index $I$ fixes the computation to be verified (\eg the layered circuit for a model), instance $x$ and witness $w$ are the public and private values, respectively.

\noindent
\textbf{Indexed Relation}~\cite{chiesa2020marlin}: An indexed relation $\mathcal{R}$ is a set of triples $(I, x, w)$ 
where $I$, $x$, $w$ are the index, instance, and witness, respectively, defined as:
\begin{align}
\mathcal{R}=\Big\{(I, x=\{q,r\}, w=\{\mathsf{adv},\theta\}) \;\big|\; C_I(x,w)=1\Big\}. \nonumber
\end{align}
The indexed language ($\mathcal{L}(\mathcal{R})$) consists of all pairs ($(I,x)$) for which there exists a witness ($w$) that satisfies the relation.

A zk-SNARK for relation $\mathcal{R}$ consists of three algorithms \\$(\mathsf{Setup}, \mathsf{Prove}, \mathsf{Verify})$~\cite{Groth16}. 
They allow \prover to show that $(I,x) \in \mathcal{L}(\mathcal{R})$ by proving knowledge of a witness $w$, without revealing $w$ to \verifier. A zk-SNARK satisfies three security properties (formal definitions in \S\ref{sec:zk-snark}): 1) \textit{Completeness}: honestly generated proofs are always accepted by \verifier; 2) \textit{Soundness}: false proofs cannot be accepted except with negligible probability; 3) \textit{Zero-knowledge}: \verifier learns nothing beyond validity of the statement (no information about $w$).

In zkML systems~\cite{zkcnn,zkgpt,zkllm},  zk-SNARKs are typically constructed using a zero-knowledge GKR protocol within a commit-and-prove framework~\cite{campanelli2019legosnark}, yielding a \emph{CP-SNARKs} (commit-and-prove SNARKs). \prover first commits (using a PCS) to the witness polynomials and then runs the zero-knowledge GKR protocol, replacing each interactive message with the corresponding committed polynomial evaluations or openings. 

\subsection{Fiat-Shamir For Non-interactive Protocols}
\label{sec:fs}

\begin{figure} [tb]
    \centering
    \includegraphics[trim={15mm 3mm 10mm 7mm}, scale=0.8]{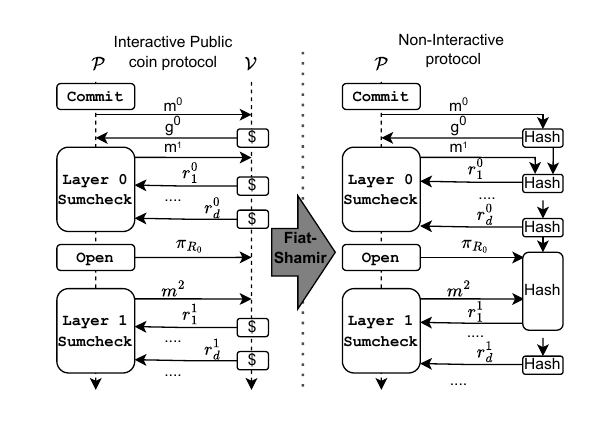}   \vspace{-0.3cm}
    \caption{Transforming an interactive public coin protocol into a non-interactive protocol using Fiat-Shamir heuristic.}
    \label{fig:fs}
\end{figure}

Fig.~\ref{fig:fs} (left) shows the interactive protocol used in zkML systems~\cite{zkgpt,zkcnn,zkllm}. First, \prover commits to the input and masking polynomials. The commitment is sent as part of $m^0$ to \verifier before starting the GKR protocol. During the protocol, \verifier samples random values ($g^{0}, r_1^0, \dots$) and sends them to \prover at specific points, \eg $g^0$ is sent only after \prover transmits $m^0$. In typical zkML deployments, such interaction is impractical.
zkML systems address this via the Fiat-Shamir heuristic~\cite{fiat1986prove}, converting the interactive protocol to a non-interactive one (Fig.~\ref{fig:fs}, right). It replaces \verifier's random sampling with a pseudo-random function, typically a hash such as SHA-3~\cite{sha3}. \prover can thus compute the challenges ($g^0, r_1^0, \dots$) locally by hashing the protocol transcript up to that point, \eg $g^0$ is obtained by hashing $m^0$. This yields a non-interactive zero-knowledge argument in the random oracle model~\cite{fiat1986prove} and is standard in zkML systems~\cite{zkcnn,zkllm,zkgpt}. Importantly, Fiat-Shamir remains sound only when \emph{every} prior prover message is included when computing the hash for each challenge~\cite{fiat1986prove, Thaler2022}, \eg $r_1^0$ must be computed by hashing $g^0$, $m^0$, and $m^1$. Thus, layer $i$'s proof cannot begin until all preceding layers' proofs are computed.

%% file: src/key_idea.tex
\section{\X: Key Ideas}
\label{sec:key_idea}

\begin{figure}[h]
\vspace{-0.75cm}
\begin{adjustbox}{max totalsize={\linewidth}{\textheight}}
\centering
\includegraphics[trim={10mm 2mm 7mm 2mm},clip]{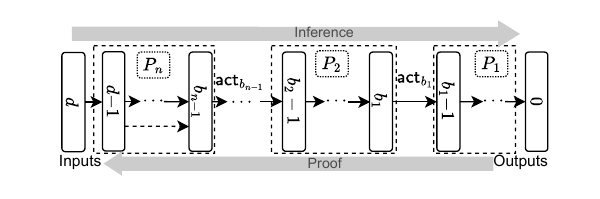}
\end{adjustbox}
\caption{Arithmetic circuit for \X with $n$ partitions.} 
\label{fig:zkmod}
\end{figure}
We accelerate zkML proof generation by exploiting the parallelism across independent sub-proofs. We divide the neural network into multiple partitions and generate a sub-proof for each. We describe our approach for partitioning along model layers; other schemes appear in \S\ref{sec:partitioning}. Each partition (indexed from output to input) consists of a subset of contiguous layers. Fig.~\ref{fig:zkmod} shows the arithmetic circuit of a zkML application with \X, where the model is split into $n$ partitions, $P_n,\cdots,P_1$. Each sub-proof shows that its partition correctly transforms its input activations into output activations using its layer operations (\eg matrix multiplication, convolution). To match the correctness guarantee of a monolithic proof, we additionally verify that the input activations of partition $P_{j-1}$ are equal to the output activations of $P_j$.


Our setting is an instance of SNARK composition, where sub-relations are linked by enforcing a relationship between shared boundary witnesses (as in LegoSNARKs~\cite{campanelli2019legosnark}). Linking sub-proofs in zkML poses two challenges. \textit{(1) Proving equality via polynomial commitments.} With a separate sub-proof per partition, the boundary activations are private witnesses, so each sub-proof commits to its own (masked) polynomials. Partition $P_{j-1}$'s input activations are committed together with its model parameters as a single polynomial. Thus, composing sub-proofs requires showing that the activations in this committed input polynomial are identical to $P_j$'s committed (masked) output polynomial, using only the commitments. \textit{(2) Preserving zero-knowledge at boundaries.} Since intermediate activations are private witnesses, \X must enforce equality of the boundary values across consecutive sub-proofs while \emph{revealing no additional information} beyond what the public values imply.

\X addresses both challenges by linking partitions through \emph{shared (masked) commitments} to output activations of the boundary layer in each partition. \prover creates a single PCS commitment to a partition's output activations and reuses it in every sub-proof that consumes them as inputs. The input-layer polynomial of partition $i$ is constructed such that its input activations are derived from (and checked against) this shared masked polynomial (see \S\ref{sec:boundary-poly}). Due to the binding property of the PCS, any correctly opened values must correspond to the same underlying boundary polynomial, thereby establishing equality of the polynomials. For zero knowledge, \X masks each boundary polynomial with a fresh random polynomial (see \S\ref{sec:zk-boundary}). The number of variables in the masking polynomial is chosen such that it masks 
the total number of evaluations of the corresponding polynomial across all sub-proofs. 

Prior proof-composition techniques~\cite{campanelli2019legosnark,campanelli2021lunar} use a separate zk-SNARK to relate two committed polynomials. LegoSNARKs~\cite{campanelli2019legosnark} formalizes modular composition via CP-SNARKs, providing linking gadgets that prove equality between two polynomial commitments. These gadgets assume both sub-proofs commit to the \emph{same} polynomial encoding of the boundary values. In zkML systems~\cite{zkgpt,zkcnn,zkllm}, however, the polynomials are committed as low-degree extensions (LDEs) masked by independent random polynomials per sub-proof. Thus, the committed polynomials differ even when the underlying witness is identical, making direct use of LegoSNARK's gadget non-trivial. Lunar~\cite{campanelli2021lunar} offers a comparable mechanism but targets univariate polynomials, whereas zkML systems~\cite{zkcnn,zkgpt,zkllm,qu_verfcnn_2025} typically operate over \emph{multivariate} polynomials, requiring a redesign to support multivariate equality checks.

\noindent
\textbf{Formal description of \X:}
Recall from \S\ref{sec:background} that existing zkML systems~\cite{zkcnn, zkgpt, zkllm} prove the relation:
\[
\mathcal{R}_{\mathsf{zkml}} = \Big\{ (I, x = \{q,r\}, w = \{\mathsf{adv}, \theta\}) \;\big|\; C_I(x,w) = 1 \Big\},
\]
where the index $I$ is the circuit, the public values $x$ consist of the query $q$ and response $r$. Here, the witness $w$ includes the model parameters $\theta$ and advice values $\mathsf{adv}$. \X generates proofs for partitions $P_1, \ldots, P_n$, defined by $n-1$ intermediate boundary layers $b_1 < \cdots < b_{n-1}$, each $b_j \in \{1,\ldots,d-1\}$. Setting $b_n = d$ (input layer) and $b_0 = 0$ (output layer), partition $P_j$ consists of layers $\{b_{j-1}, \ldots, b_{j}-1\}$ for $j \in \{1, \ldots, n\}$ (so $P_n$ includes the input layer and $P_1$ the output layer). The output of layer $b_{j-1}$, denoted $\mathsf{act}_{b_{j-1}}$, is the output of $P_{j}$ and the input to $P_{j-1}$. For ease of exposition, we use a common indexing for the ML model and its layered arithmetic circuit, and assume each circuit layer is verified via sumcheck applied to Equation~\eqref{eq:modified-gkr-wiring}. In practice, some layers are instead verified via lookup protocols~\cite{zkcnn, zkllm} (\eg Lasso~\cite{lasso}), which are themselves sumcheck-based; \S\ref{sec:protocol} describes \X's integration with such frameworks. The query $q$ is a public input in $P_n$ and the response $r$ is the public output of $P_1$. The model parameters are split into $\theta_1, \ldots, \theta_n$, each a private input to the proof of the corresponding partition, and the boundary activations are private witnesses in the two adjacent proofs. Formally, the indexed relations for the $n$ sub-proofs are:
\begin{align*}
\mathcal{R}_{1} &= \Big\{ (I_1, x_1 = \{r\}, w_1 = \{ w_{b_{1}} , \mathsf{adv}_1, \theta_1\}) \\
&\quad\;\big|\; C_{I_1}(x_1,w_1) = 1 \Big\}, \\
\mathcal{R}_{j} &= \Big\{ (I_j, x_j = \{\}, w_j = \{\mathsf{act}_{b_{j-1}}, w_{b_{j}}, \mathsf{adv}_j, \theta_j\}) \\
& \quad \;\big|\; C_{I_j}(x_j,w_j) = 1 \Big\}, \quad \text{for } j \in \{2, \ldots, n-1\}, \\
\mathcal{R}_{n} &= \Big\{ (I_n, x_n = \{q\}, w_n = \{\mathsf{act}_{b_{n-1}},\mathsf{adv}_n, \theta_n\}) \\
& \quad\;\big|\; C_{I_n}(x_n,w_n) = 1 \Big\}.
\end{align*}
If a prover produces accepting proofs for $\mathcal{R}_1, \ldots, \mathcal{R}_n$ and additionally shows $w_{b_j} = \mathsf{act}_{b_j}$ for all boundary layers $j \in \{1, \ldots, n-1\}$, this is equivalent to a proof for $\mathcal{R}_{\mathsf{zkml}}$~\cite{campanelli2019legosnark}.

\subsection{Proving Equality Using Polynomial Commitments: Boundary Layer Construction}
\label{sec:boundary-poly}
To verify equality between adjacent partitions $P_j$ and $P_{j+1}$, we must show that the output polynomial of $P_{j+1}$ equals a portion of the input polynomial of $P_j$. It equals only a portion because, in GKR-based zkML, the input polynomial incorporates both activation/query values and model weights. We achieve this by constructing an input polynomial for $P_j$ that distinctly separates these components. Let $\tilde{V}_{b_j}$ be the boundary polynomial for the output $\mathsf{act}_{b_j}$ of partition $P_{j+1}$ (where $1 \leq j < n$). For $P_j$, the input layer polynomial $\tilde{V}_{\mathsf{inp},j}$ combines $P_{j+1}$'s output polynomial $\tilde{V}_{b_j}$ with the model-weight polynomial $\tilde{V}_{\theta_j}$ for $P_j$. We combine them using chunking~\cite{zkgpt,gkr_chunking} via a selector variable in $\tilde{V}_{\mathsf{inp},j}$. Let
$$k_j = \log_2\left(2 \cdot \max(|\mathsf{act}_{b_j}|, |\theta_j|)\right),$$
where $(r_1, \ldots, r_{k_j-1})$ index positions in either $\tilde{V}_{b_j}$ or $\tilde{V}_{\theta_j}$. The input polynomial for $P_j$ is
\begin{align}
\tilde V_{\mathsf{inp},j}(r_1,\ldots,r_{k_j}) = &(1 - r_{k_j})\,\tilde V_{b_{j}}(r_1,\ldots,r_{k_j-1}) \nonumber\\
&+ r_{k_j}\,\tilde V_{\theta_j}(r_1,\ldots,r_{k_j-1}),
\label{eq:boundary-poly}
\end{align}
where $r_{k_j}$ selects which polynomial is evaluated. For each boundary layer $b_j$, \prover commits once to $\tilde V_{b_j}$ in the proof for $P_{j+1}$ (where $b_j$ is the output layer) and reuses it in the proof for $P_j$ (where $b_j$ is the input layer). \prover commits to $\tilde V_{\theta_j}$ separately in $P_j$'s proof. In $P_{j+1}$'s proof, \verifier queries $\tilde V_{b_j}$ directly. At the end of $P_j$'s proof, when performing oracle access to $\tilde V_{\mathsf{inp},j}$, \verifier queries $\tilde V_{b_{j}}$ and $\tilde V_{\theta_j}$ separately and computes $\tilde V_{\mathsf{inp},j}$ via Eq.~\ref{eq:boundary-poly}. By the binding property of the PCS, \prover cannot use $\textsf{PCS.Open}$ on the commitment of $\tilde V_{b_{j}}$ to prove evaluations of two different polynomials. Consequently, the values $w_{b_{j}}$ in $P_j$'s proof match the activations $\mathsf{act}_{b_{j}}$ from $P_{j+1}$'s proof. This enforces boundary-activation equality without any additional linking proofs, while remaining fully compatible with existing zkML systems~\cite{zkcnn, zkllm, zkgpt}.

\noindent
\textbf{Shared model weights between partitions.}
Two partitions $P_k$ and $P_j$ ($j<k$) may share weights when different layers reuse the same parameters. The shared values are included as input weight $\theta$ in the sub-proof for the partition that uses them first during inference, \ie $P_k$. They are then encoded, together with $P_k$'s output, into $P_k$'s boundary polynomial. Using this boundary polynomial as input to $P_j$ ensures both sub-proofs operate on the same weight values.

\subsection{Zero-Knowledge At Partition Boundaries}
\label{sec:zk-boundary}

Using LDEs of multilinear polynomials as in Eq.~\eqref{eq:lde} suffices to ensure the zero-knowledge properties of all polynomials in a sub-proof for partition $P_{j+1}$ except the output polynomial $\tilde{V}_{b_j}$. In \X, \verifier obtains oracle access to $\tilde V_{b_j}$ in at least two proofs: in $P_{j+1}$'s proof, to query the output of layer $b_j$, and in $P_j$'s proof, to verify the input-layer construction (Eq.~\ref{eq:boundary-poly}). Thus, without proper masking, these evaluations of the boundary polynomials reveal information about the underlying activations.

The number of evaluations of each boundary polynomial depends \emph{only} on the model architecture and the chosen partition boundaries. \prover can therefore compute it in advance. For the partitioning in Fig.~\ref{fig:zkmod}, every boundary polynomial is evaluated exactly twice for each boundary $b_j$: once in $P_j$ and once in $P_{j+1}$. This count can grow with certain operations, such as skip connections. Suppose a skip connection runs from layer $d-1$ to boundary $b_{n-1}$, with the two layers in different, non-contiguous partitions. Then the polynomial for layer $d-1$ must be evaluated three times, once in each of the three sub-proofs that depend on it. The third evaluation occurs in the partition containing $b_{n-1}$. In general, a boundary polynomial is evaluated as many times as the number of partitions that take it as input.

To preserve zero-knowledge, \prover selects the number of random variables in $R_{b_j}$ to mask \emph{all} evaluations of the boundary polynomial across all sub-proofs. \prover then replaces each $\tilde{V}_{b_j}$ with its LDE. For example, when $\tilde{V}_{b_j}$ is evaluated twice, $R_{b_j}$ includes monomials over two variables (Eq.~\eqref{eq:masking-2}). We discuss how to adapt $R_b$ for other cases in \S\ref{sec:mult_part}. Thus, Eq.~\ref{eq:boundary-poly} becomes:
\begin{align}
\dot V_{\mathsf{inp},j}(r_1,\ldots,r_{k_j}) = &(1 - r_{k_j})\,\dot V_{b_{j}}(r_1,\ldots,r_{k_j-1}) +\nonumber\\& r_{k_j}\,\dot V_{\theta_j}(r_1,\ldots,r_{k_j-1}),
\label{eq:zk-boundary-poly}
\end{align}
where $\dot V_{b_{j}}$ and $\dot V_{\theta_j}$ are the LDEs of $\tilde V_{b_{j}}$ and $\tilde V_{\theta_j}$. The LDEs agree with the original polynomials on all Boolean inputs, while masking the true values at all verifier-chosen evaluation points. This is sufficient to ensure zero-knowledge as shown in \S\ref{sec:zk-proof}.

%% file: src/prot_desc.tex
\section{\X: Putting It All Together}
\label{sec:protocol}
\label{sec:fs_zkmod}






Consider a model decomposed into two partitions at layer~$b$: $A$ (layers $\{0,\cdots,b-1\}$) and $B$ (layers $\{b,\cdots,d\}$). The protocol proceeds in three phases.

\emph{Boundary commitment.}
\prover executes the forward computation up to layer $b$ to obtain the intermediate activations $\mathsf{act}_b$. \prover then constructs the LDE $\dot V_b$ of the multilinear extension encoding $\mathsf{act}_b$ and commits to it using a polynomial commitment scheme. This commitment fixes the boundary values used consistently across both sub-proofs.

\emph{Proof for part $A$.}
\prover executes the remaining forward computation to obtain the intermediate activations and the response $r$. Using the committed boundary polynomial $\dot V_b$, \prover constructs the input polynomial for part $A$ as in §\ref{sec:boundary-poly}, replacing all multilinear extensions with their LDEs (§\ref{sec:zk-boundary}). \prover then executes the zero-knowledge GKR protocol for layers $\{0,\ldots,b-1\}$. At the end, the commitment of $\dot V_b$ is opened to prove the claim for the activations in the input of part $A$. This yields a zero-knowledge proof for $\mathcal{R}_A$, \ie part $A$ correctly computes the public response $r$ from the committed boundary activations and private weights $\theta_A$.

\emph{Proof for part $B$.}
 \prover executes the zero-knowledge GKR protocol for layers $\{b, \dots, d\}$, using the committed polynomial $\dot{V}_b$ as the output polynomial of layer $b$. At the start, \verifier selects a random evaluation point. \prover evaluates $\dot{V}_b$ and opens the commitment to $\dot{V}_b$ at that point. This evaluation is used as the claimed evaluation of layer $b$ for the protocol.

\verifier accepts if and only if both sub-proofs verify and all commitment openings are valid. Since \X modifies only the proof's structure and not the underlying GKR or commitment mechanisms, it applies directly to existing zkML frameworks without any modification to the commitment scheme. Protocol~\ref{prot:zkmod} gives a formal description.

\begin{figure*}[!h]
  \centering
\begin{protocol}{\X}
\label{prot:zkmod}
Let $C$ be a layered arithmetic circuit with layers indexed from $d$ (input) to $0$ (output), implementing a computation with parameters $\theta$. Let $q$ be the public query, $r$ the public response, and let $b \in \{1,\ldots,d-1\}$ denote the partition point. The prover \prover and verifier \verifier execute the following protocol.

\textbf{Setup.}
Public parameters $pp$ for the PCS and the zero-knowledge GKR protocol are generated as in §\ref{sec:background}.

\textbf{Boundary Commitment.}
\begin{enumerate}\itemsep0pt
    \item Upon receiving $q$, \prover evaluates the circuit up to layer $b$ and computes the MLE $\tilde V_b$ encoding $\mathsf{act}_b$.
    \item \prover constructs the LDE $\dot V_b(z)=\tilde V_b(z)+Z(z)\sum_{w\in\{0,1\}}R_b(z_1,w)$, where $Z(\cdot)$ and $R_b(\cdot)$ are defined in §\ref{sec:background}.
    \item \prover computes $\mathsf{com}_b\leftarrow\mathsf{PCS.Commit}(pp,\dot V_b, r_{V_b})$ and $\mathsf{com}_{R_b}\leftarrow\mathsf{PCS.Commit}(pp,R_b,r_{R_b})$. 
    \item \prover sends $(\mathsf{com}_b,\mathsf{com}_{R_b})$ to \verifier.
\end{enumerate}

\textbf{Sub-protocol A (Layers $0$ to $b{-}1$).}
\begin{enumerate}\itemsep0pt
    \item \prover computes $\mathsf{com}_{\theta_A}\leftarrow\mathsf{PCS.Commit}(pp,\dot V_{\theta_A}, r_{V,\theta_A})$ and sends $\mathsf{com}_{\theta_A}$ to \verifier.

    \item \prover and \verifier execute the zero-knowledge GKR protocol for layers $\{0,\ldots,b{-}1\}$ using the boundary input polynomial $\dot V_{\mathsf{inp},A}$ constructed from $\dot V_b$ and $\dot V_{\theta_A}$ (Eq.~\ref{eq:zk-boundary-poly}).

    \item At the end of the protocol, \verifier samples a random evaluation point $
    g^{(\mathsf{inp,A})} \xleftarrow{\$} \mathbb{F}$ for the input layer of part~$A$, and obtains the claimed value
    $\dot V_{\mathsf{inp},A}\!\left(g^{(\mathsf{inp,A})}\right)$
    from the final round of the sumcheck protocol. Let 

    \item \prover opens the corresponding commitments by sending $
    (\dot V_b(g^{(\mathsf{inp,A})}),\pi_{\mathsf{act}})
    \leftarrow
    \mathsf{PCS.Open}(pp,\dot V_b,g^{(\mathsf{inp,A})}, r_{V_b})$ and $
    (\dot V_{\theta_A}(g^{(\mathsf{inp,A})}),\pi_{\theta_A})
    \leftarrow
    \mathsf{PCS.Open}(pp,\dot V_{\theta_A},g^{(\mathsf{inp,A})}, r_{V,\theta_A})$.

    \item \verifier checks $
    \mathsf{PCS.Verify}(pp,\mathsf{com}_b,g^{(\mathsf{inp,A})},\dot V_b(g^{(\mathsf{inp,A})}),\pi_{\mathsf{act}})$ and $\mathsf{PCS.Verify}(pp,\mathsf{com}_{\theta_A},g^{(\mathsf{inp,A})},\dot V_{\theta_A}(g^{(\mathsf{inp,A})}),\pi_{\theta_A})$, 
    and aborts if either check fails.

    \item \verifier locally computes $\dot V_{\mathsf{inp},A}\!\left(g^{(\mathsf{inp,A})}\right)$
    from $\dot V_b(g^{(\mathsf{inp,A})})$ and $\dot V_{\theta_A}(g^{(\mathsf{inp,A})})$ according to Eq.~\ref{eq:zk-boundary-poly}, and checks that it matches the value obtained at the end of the sumcheck protocol.
\end{enumerate}

\textbf{Sub-protocol B (Layers $b$ to $d$).}
\begin{enumerate}\itemsep0pt
    \item \prover samples the random masking polynomials required for zero-knowledge in part~$B$. This includes the masking polynomials for the MLEs of layers $\{d,\ldots,b{+}1\}$, denoted $\{R_d,\ldots,R_{b+1}\}$, as well as the masking polynomials for the subsets of the input layer used by these layers, denoted $\{R_{d-1,d},\ldots,R_{b,d}\}$. \prover computes commitments to all these polynomials and sends the commitments to \verifier.

    \item \verifier samples a random challenge $g^{(b)} \xleftarrow{\$} \mathbb{F}$ and sends it to \prover.

    \item \prover opens the boundary commitment by sending $
    (\dot V_b(g^{(b)}), \pi_b) \leftarrow \mathsf{PCS.Open}(pp,\dot V_b,g^{(b)}, r_{V_b})$.

    \item \verifier checks $ \mathsf{PCS.Verify}(pp,\mathsf{com}_b,g^{(b)},\dot V_b(g^{(b)}),\pi_b)$, 
    and aborts if the check fails.

    \item \prover and \verifier execute a zero-knowledge sumcheck protocol to verify the following claim for layer $b$:
    \begin{align}
\dot{V}_b(g^{(b)}) 
=&\sum_{x \in \{0,1\}^{s_{b+1}}} 
\widetilde{\mathrm{add}}_{b,b+1}(g^{(b)},x) \cdot \dot{V}_{b+1}(x) +\sum_{x,y \in \{0,1\}^{s_d}} 
\widetilde{\mathrm{mult}}_{b,d,d}(g^{(b)},x,y) 
\cdot \tilde{V}_{b,d}(x) \cdot \dot{V}_{b,d}(y) \nonumber\\
&+\sum_{x \in \{0,1\}^{s_d}} 
\widetilde{\mathrm{add}}_{b,d}(g^{(b)},x) \cdot \dot{V}_{b,d}(x) 
+\sum_{x,y \in \{0,1\}^{s_{b+1}}} 
\widetilde{\mathrm{mult}}_{b,b+1}(g^{(b)},x,y) 
\cdot \dot{V}_{b+1}(x) \cdot \tilde{V}_{b+1}(y)  \nonumber\\ 
&+\sum_{\substack{
x \in \{0,1\}^{s_{b+1}} \\
y \in \{0,1\}^{s_d}
}}
\widetilde{\mathrm{mult}}_{b,b+1,d}(g^{(b)},x,y) 
\cdot \tilde{V}_{b+1}(x) \cdot \dot{V}_{b,d}(y)+\; Z_b(g^{(b)})\sum_{w\in\{0,1\}} R_b\!\big(g^{(b)}_1,w\big).
    \label{eq:vb_claim}
    \end{align}

    \item At the conclusion of the zero-knowledge sumcheck, \verifier obtains two claims $\dot V_{b+1}(u^{(b+1)})$ and $\dot V_{b+1}(v^{(b+1)})$. \prover opens the corresponding masking polynomial $R_b$ at $g^b_1,c$ for a verifier chosen random $c$, and \verifier checks that the claimed evaluations are consistent with the value from the oracle.
    \item \prover and \verifier continue executing the zero-knowledge GKR protocol for layers $\{b{+}1,\ldots,d\}$.
\end{enumerate}

\textbf{Decision.}
\verifier accepts iff both sub-protocols accept and all commitment verifications succeed.
\end{protocol}
\end{figure*}

The \X construction extends recursively to partition a model into more than two parts. Each additional partition introduces a new boundary commitment shared by adjacent sub-proofs. This decomposition reduces the size of each sub-circuit, lowering the memory required per sub-proof. When sub-proofs are generated sequentially, only one partition's witness values and prover instance must be stored at a time, reducing peak prover memory. When generated in parallel, \X instead reduces end-to-end proof generation time, subject to available hardware resources.

\noindent
\textbf{Using \X with lookup-based layers.}
Modern zkML systems such as zkGPT~\cite{zkgpt} and zkLLM~\cite{zkllm} use lookup arguments (\eg Lasso~\cite{lasso}) to prove activation functions and normalization layers. These lookup arguments are themselves implemented using sumcheck-based protocols and polynomial commitments. \X applies directly: when a partition boundary is a lookup-based layer, its boundary activations are treated \emph{in the same way} as standard arithmetic layer outputs. The committed boundary polynomial encodes the activation values, while the lookup constraints are enforced entirely within the sub-proof containing the corresponding layers. \X therefore requires no modification to existing lookup protocols and introduces no assumptions beyond those of existing zkML frameworks~\cite{zkllm,zkgpt}.

\X is a general construction: although primarily designed for zkML
systems that use GKR, it can be used to prove arbitrary computations
(\S\ref{sec:other_computation}) and with zkML systems that use
table-based arithmetic circuits (\S\ref{sec:plonk}).

\noindent
\textbf{Using Fiat--Shamir For Non-interactive Proofs.} To obtain a non-interactive protocol, we apply the Fiat--Shamir transformation (Section~\ref{sec:fs}) to each sub-protocol in the random-oracle model. Recent work~\cite{fs_attack} shows that Fiat--Shamir security depends on the circuit implementation. In particular, the circuit depth must be less than the sum of the computational depths required to compute the hash function and the PCS. The per-partition GKR sub-protocols in \X remain unchanged from the configurations analyzed in the underlying CP-SNARK constructions~\cite{libra,zkcnn,zkgpt}, which also use Fiat--Shamir compilation in the random-oracle model. We therefore assume the underlying CP-SNARKs satisfy this depth condition, \ie they are secure in the random-oracle model.

In \X, partitions are linked only through shared $\PCS$ commitments and their openings, \ie the Fiat--Shamir transcripts for each partition are independent. Thus, the depth condition applies independently to each sub-protocol and is unaffected by the number of partitions $n$. If an underlying circuit does not satisfy this condition, the partitioning scheme must ensure that each partition's depth remains below the maximum allowed for the corresponding hash function and commitment scheme used.

\subsection{Partitioning The Model}
\label{sec:partitioning}
The choice of partition boundaries is determined by the available hardware resources (\eg number of threads and memory) and by the structure of the proof circuit for each operation. The cost of each partition depends on the operations it contains, the associated kernel sizes, the model and input dimensions. We describe the three partitioning strategies we evaluate in \S\ref{sec:evaluation}.

\emph{Partitioning along layer boundaries in GPT-2~\cite{gpt}.} We use the same arithmetization scheme as zkGPT~\cite{zkgpt} and place partition boundaries between the transformer decode blocks. Thus, using this strategy, we can obtain a maximum of $12$ partitions (as there are 12 transformer blocks in GPT-2). We apply the circuit optimizations proposed by zkGPT~\cite{zkgpt} within each partition.

\emph{Partitioning along input tokens in GPT-2~\cite{gpt}.} For LLMs, \X can also partition along the input tokens. The key--value (KV) activations computed for a token are used to compute the attention scores of every subsequent token. We therefore commit to the intermediate KV values produced within each token partition, and subsequent partitions consume these commitments as advice in their input layer. The shared-commitment mechanism is used between partitions (\S\ref{sec:key_idea}) to ensure consistency between the advice KV values and the values produced by earlier partitions. Thus, combining the partitions along the input tokens and the layer boundaries allows additional proof computations to be run in parallel, thereby decreasing prover time.

\emph{Partitioning along layer boundaries in CNNs.} We use the same arithmetization as zkCNN~\cite{zkcnn} and place partition boundaries immediately after each convolution block. Here, a convolution block consists of a sequence of convolution layers followed by a max-pooling layer. Thus, we obtain a maximum of $5$ partitions for the models we evaluation.  

The choice of partition boundaries and count $K$ controls a memory-compute trade-off for the prover. Increasing $K$ reduces per-partition memory and can reduce prover time when sub-proofs run in parallel, but it also introduces additional boundary commitments and reduces the number of threads available to each sub-proof. Within a single partition scheme, three failure modes can degrade performance: 1) an uneven split leaves the largest sub-proof on the critical path; 2) fine-grained partitions incur commitment and opening overhead that can offset the performance improvements from \X; and 3) when the partitions are run in parallel the total memory required can be above the monolithic baseline.

For GPT-2, the transformer blocks are structurally identical, so any partition count yields equal-cost sub-proofs. For CNNs, where per-layer cost varies, \X resolves the partitioning problem by exhaustively profiling the feasible configurations offline and selecting the one with the lowest prover time. This is a feasible approach since the profiling is only a one-time cost for deploying any given model. The same procedure applies to any new model: choose partition boundaries aligning with the natural structural units (\eg transformer, convolution, or residual blocks). The boundaries must be chosen such that constraint-merging optimizations (similar to zkGPT~\cite{zkgpt}) can be applied within each partition. Offline profiling can be used to choose the best configuration for deployment if the partitions are not balanced. For models with decomposable computation along the inputs (\eg LLMs), input-dimension partitioning adds further parallelism.

%% file: src/security_analysis.tex
\section{Security Properties Of \X}
\label{sec:zk-proof}
Here, we prove that Protocol~\ref{prot:zkmod} satisfies the completeness, knowledge soundness, and zero-knowledge properties of a zk-SNARK for $\Rzkml$ (formal definitions in \Cref{def:zksnark}). Arguments for multiple partitions obtained by recursive application of \X follow similarly. Intuitively, \X splits the overall computation into sub-proofs and instantiates a standard zk-SNARK for each. We state the completeness and knowledge soundness theorem below, which follows directly from the underlying zk-SNARK construction and the PCS~\cite{campanelli2019legosnark}. We then sketch the proof that Protocol~\ref{prot:zkmod} is zero-knowledge. We discuss the asymptotic complexity and the expanded proofs in \S\ref{sec:proof}.

\begin{theorem}[Knowledge Argument of \(\Pi_{\zkmod}\)]
\label{thm:knowledge-argument-twopart}
Let \(\Pi_{\zkmod}=(\mathsf{Setup},\mathsf{Prove},\mathsf{Verify})\)
be the \X protocol for two partitions as shown in Protocol~\ref{prot:zkmod}, for the
indexed relation \(\Rzkml\). Let the two sub-proofs correspond to the indexed
relations \(\mathcal{R}_A\) and \(\mathcal{R}_B\), where \(\mathcal{R}_A\)
proves correctness of layers \(\{0,\ldots,b-1\}\) and \(\mathcal{R}_B\)
proves correctness of layers \(\{b,\ldots,d\}\). Suppose that:
\begin{enumerate}
    \item \(\PCS\) satisfies completeness and binding
    (Definition~\ref{def:pcs});
    \item the GKR-based CP-SNARKs underlying the sub-proofs for
    \(\mathcal{R}_A\) and \(\mathcal{R}_B\) satisfy completeness and
    knowledge soundness (Definition~\ref{def:zksnark}).
\end{enumerate}
Then \(\Pi_{\zkmod}\) satisfies completeness and knowledge soundness for
\(\Rzkml\). Equivalently, \(\Pi_{\zkmod}\) is a knowledge argument for
\(\Rzkml\).
\end{theorem}

\begin{theorem}[Zero-Knowledge of \(\X\)] \label{thm:zk} Let \(\X\) be the two-partition protocol in Protocol~\ref{prot:zkmod}, with a single boundary layer \(b\), and let \(\mathcal{R}_{\mathsf{zkml}}\) be the original indexed relation. Suppose that: \begin{enumerate} \item the GKR-based CP-SNARK sub-proofs for \(\mathcal{R}_A\) and \(\mathcal{R}_B\) are zero-knowledge (Definition~\ref{def:zksnark}), with simulator \(\mathcal{S}_{\mathsf{sc}}\); \item \(\PCS\) satisfies hiding (Definition~\ref{def:pcs}), with simulator \(\mathcal{S}_{\mathsf{pc}}\). \end{enumerate} Then \(\X\) is zero-knowledge for \(\mathcal{R}_{\mathsf{zkml}}\). That is, there exists a \(\PPT\) simulator \(\mathcal{S}_{\mathsf{mod}}\) such that for every $(\mathsf{i},\mathsf{x},\mathsf{w})\in\mathcal{R}_{\mathsf{zkml}}$ and every \(\PPT\) distinguisher \(\mathcal{D}\), \[ \left| \Pr[\mathcal{D}(\pi)=1] - \Pr[\mathcal{D}(\pi')=1] \right| \leq \negl(\lambda), \] where \[ \mathsf{pp}\leftarrow\mathsf{Setup}(1^\lambda,\mathsf{i}), \pi\leftarrow\mathsf{Prove}(\mathsf{pp},\mathsf{i},\mathsf{x},\mathsf{w}), \] and \[ \pi'\leftarrow \mathcal{S}_{\mathsf{mod}}(\mathsf{pp},\mathsf{i},\mathsf{x}). \] \end{theorem}

\begin{proof}[Proof sketch]
We construct a simulator \(\mathcal{S}_{\mathsf{mod}}\) by combining the
simulators for the two underlying components: the simulator
\(\mathcal{S}_{\mathsf{sc}}\) for the GKR-based CP-SNARK transcripts and the simulator \(\mathcal{S}_{\mathsf{pc}}\) for the polynomial commitments.

All transcript components internal to the two sub-protocols are simulated
exactly as in the standard commit-and-prove GKR zero-knowledge argument. Thus,
by the zero-knowledge of the sub-proofs and the hiding of the PCS, these parts
are indistinguishable from a real execution~\cite{libra}.

In \X, the two sub-protocols share the boundary
polynomial \(\dot V_b\). Thus, the verifier sees two evaluations of this polynomial,
at query points \(u=g^{(b)}\) and \(v=g^{(\mathsf{inp},A)}\), together with one
evaluation of the boundary masking polynomial \(R_b(u_1,c)\). The boundary mask
is chosen so that these three revealed values are three independent linear
combinations of the random masking coefficients, \ie the corresponding linear map has a full
rank. Hence the revealed values are distributed uniformly over
\(\mathbb{F}_p^3\) and are statistically independent of the boundary
polynomial \(\widetilde V_b\).

\(\mathcal{S}_{\mathsf{mod}}\) samples these three boundary-related
values uniformly, uses \(\mathcal{S}_{\mathsf{pc}}\) to simulate the relevant
commitments and openings, and uses \(\mathcal{S}_{\mathsf{sc}}\) to simulate
the GKR/sumcheck transcripts consistently with the sampled values. By a standard hybrid argument, the simulated transcript has the same distribution on the boundary messages as the real transcript. Therefore, \X satisfies zero-knowledge.
\end{proof}

%% file: src/evaluation.tex
\section{Evaluation}
\label{sec:evaluation}

\subsection{Implementation}
\textbf{Software:} We evaluate \X by integrating it into two modern zkML frameworks: zkCNN~\cite{zkcnn} and zkGPT~\cite{zkgpt}. \X leaves the underlying arithmetization and the PCS unchanged. Accordingly, we use Hyrax~\cite{hyrax} as the PCS for both boundary and input polynomials. The commitment kernels are accelerated via Pippenger's algorithm~\cite{pippenger1976evaluation}, while the sumcheck kernels leverage multi-threaded parallelism. All field and curve operations use the mcl library~\cite{mcl} over the BLS12-381~\cite{bls12381} and BN254~\cite{bn254} curves. We optimized the implementation to dynamically allocate the necessary memory for intermediate computations in the prover and release it immediately after use.



\textbf{Metrics.} We compare the impact of \X using the following metrics: prover time, peak memory, verifier time, and response time. The prover time includes the duration to commit to all witness values, \eg model weights and activations, to generate the proof by executing the GKR protocol, and opening the respective commitments. Peak memory denotes the total memory required to store witness values and intermediate values during proof generation, \eg tables used to compute sumcheck protocol messages via dynamic programming. The verifier time is the duration for the verifier to execute the verification algorithm. The response time is calculated as the prover time plus the transmission time for the proof over a 100 MBps network. We also report speedups by normalizing the prover/verifier/response time of \X with the respective baselines.

\textbf{Hardware:} All experiments were conducted on a system equipped with dual-socket AMD EPYC{\texttrademark}~9654 processors, featuring 192 cores, and 756 GB of main memory in total. In our experiments, the $K$ sub-proofs are executed in parallel using a total of 192 threads (one thread per core), \ie each sub-proof uses $\lfloor192/K\rfloor$ threads. Thus, \X and the baseline use \emph{identical} hardware resources. Although our experiments were conducted using CPUs, \X can be implemented directly in GPU-based frameworks. 

\subsection{Performance of \X on CNNs}
\begin{table}[t]
\centering\small
\setlength{\tabcolsep}{3pt}
\resizebox{\columnwidth}{!}{%
\begin{tabular}{c l c c c c c}
\toprule
Model & Scheme
& \begin{tabular}[c]{@{}c@{}}Prover Time\textsuperscript{\textdagger} (s)\\(Speedup)\end{tabular}
& \begin{tabular}[c]{@{}c@{}}Verifier \\ Time (s)\end{tabular}
& \begin{tabular}[c]{@{}c@{}}Proof\\ Size (KB)\end{tabular}
& \begin{tabular}[c]{@{}c@{}}Memory (GB)\\(Baseline/\X)\end{tabular} \\
\midrule
\multirow{5}{*}{AlexNet}
 & zkCNN
 & 24.631 {\scriptsize ($1.00\times$)}
 & 0.0010
 & 273.0
 & \textbf{4.14} {\scriptsize\boldmath ($1.00\times$)} \\
 & \X ($K{=}2$)
 & 18.226 {\scriptsize ($1.35\times$)}
 & 0.0011
 & 376.8
 & 8.99 {\scriptsize ($0.46\times$)} \\
 & \X ($K{=}3$)
 & 13.534 {\scriptsize ($1.82\times$)}
 & 0.0008
 & 575.4
 & 8.58 {\scriptsize ($0.48\times$)} \\
 & \X ($K{=}4$)
 & 12.885 {\scriptsize ($1.91\times$)}
 & \textbf{0.0008}
 & 726.8
 & 10.02 {\scriptsize ($0.41\times$)} \\
 & \X ($K{=}5$)
 & \textbf{9.837} {\scriptsize\boldmath ($2.50\times$)}
 & 0.0012
 & 734.1
 & 11.32 {\scriptsize ($0.37\times$)} \\
\midrule
\multirow{5}{*}{AlexNet-Wide}
 & zkCNN
 & 83.852 {\scriptsize ($1.00\times$)}
 & 0.0015
 & 470.0
 & \textbf{16.37} {\scriptsize\boldmath ($1.00\times$)} \\
 & \X ($K{=}2$)
 & 72.251 {\scriptsize ($1.16\times$)}
 & 0.0012
 & 670.3
 & 27.91 {\scriptsize ($0.59\times$)} \\
 & \X ($K{=}3$)
 & 53.718 {\scriptsize ($1.56\times$)}
 & 0.0007
 & 869.9
 & 33.75 {\scriptsize ($0.49\times$)} \\
 & \X ($K{=}4$)
 & 45.697 {\scriptsize ($1.83\times$)}
 & \textbf{0.0005}
 & 1069.3
 & 34.97 {\scriptsize ($0.47\times$)} \\
 & \X ($K{=}5$)
 & \textbf{37.416} {\scriptsize\boldmath ($2.24\times$)}
 & 0.0007
 & 1268.8
 & 39.22 {\scriptsize ($0.42\times$)} \\
\midrule
\multirow{5}{*}{VGG16}
 & zkCNN
 & 66.966 {\scriptsize ($1.00\times$)}
 & 0.0023
 & 345.3
 & \textbf{9.59} {\scriptsize\boldmath ($1.00\times$)} \\
 & \X ($K{=}2$)
 & 42.091 {\scriptsize ($1.59\times$)}
 & 0.0024
 & 736.7
 & 17.30 {\scriptsize ($0.55\times$)} \\
 & \X ($K{=}3$)
 & 30.081 {\scriptsize ($2.23\times$)}
 & 0.0018
 & 935.6
 & 19.52 {\scriptsize ($0.49\times$)} \\
 & \X ($K{=}4$)
 & 24.289 {\scriptsize ($2.76\times$)}
 & 0.0012
 & 943.0
 & 16.67 {\scriptsize ($0.58\times$)} \\
 & \X ($K{=}5$)
 & \textbf{20.604} {\scriptsize\boldmath ($3.25\times$)}
 & \textbf{0.0008}
 & 1045.9
 & 17.08 {\scriptsize ($0.56\times$)} \\
\bottomrule
\end{tabular}%
}
\caption{Performance of \X on CNN workloads. \textsuperscript{\textdagger}Response time $\approx$ Prover time in all cases.}
\label{tab:cnn_split}
\end{table}

Table~\ref{tab:cnn_split} presents the performance of \X for three CNNs in comparison to the baseline (zkCNN~\cite{zkcnn}). We evaluate AlexNet~\cite{alexnet}, AlexNet-Wide~\cite{alexnet} (with $4\times$ the channels), and VGG16~\cite{vgg16}. We highlight four observations.

First, \X reduces prover time for all three CNN workloads. For example, in AlexNet, prover time decreases from $24.63$s to $9.84$s with $K{=}5$, giving a $2.50\times$ speedup. Similarly, \X achieves a $2.24\times$ and $3.25\times$ speedup for AlexNet-Wide and VGG16, respectively. Thus, each sub-proof is executed in parallel, allowing the proofs for two layers to be computed simultaneously. This reduces the total prover time required to compute the proofs, despite the additional computation required to commit to the boundary polynomials in \X. This additional computation required by \X is $\sim1\%$ of the total prover time. 

Second, the response-time speedup closely matches the prover-time speedup despite an increase in proof size. For example, \X achieves $2.5\times$ speedup for both prover time and response time in AlexNet, when using $5$ partitions, while the proof size increases from $273$KB to $734.1$KB. This is because proof transmission time at 100 MBps is negligible in comparison to the prover time: even the largest proof in the table is only $1268.8$KB and adds about $12$ms to response time, while the prover times are several seconds. Thus, the end-to-end response time is dominated by proof generation.

Third, peak memory increases for \X in all three workloads. For example, for AlexNet peak memory grows from $4.14$GB in the baseline to $11.32$GB in \X with $5$ partitions. This is because, for \X, we instantiate the precomputed data for commitment openings for each sub-proof. Since the sizes of CNN workloads are small, the fixed memory cost of replicating the precomputed data dominates the memory utilization. However, this increase in memory is not significant, \ie the peak memory for all configurations in CNNs is less than $40$ GB.

Fourth, we observe a reduction in verifier times for \X when normalised to the respective baselines (up to $1.28\times$, $2.89\times$, and $2.72\times$ for AlexNet, AlexNet-Wide, and VGG16, respectively). However, the absolute differences are only on the order of milliseconds: verifier time is below $3$ms for all baseline and partitioned executions. This variation is negligible compared to prover time, which ranges from tens of seconds to over a minute for the larger CNNs. 

We conclude that the sub-proof parallelism exposed by \X reduces proof-generation latency for CNNs, with the speedup increasing as the number of partitions increases.

\subsection{Performance of \X on GPT-2}
\label{sec:eval_gpt}
\begin{table}[t]
\centering\small
\setlength{\tabcolsep}{4pt}
\resizebox{\columnwidth}{!}{%
\begin{tabular}{c c *{3}{c}}
\toprule
Scheme
& \begin{tabular}[c]{@{}c@{}}Proof (KB)\end{tabular} & \begin{tabular}[c]{@{}c@{}}Prover Time\textsuperscript{\textdagger} (s)\\(Speedup)\end{tabular} & \begin{tabular}[c]{@{}c@{}}Verifier Time (s)\\(Speedup)\end{tabular} & \begin{tabular}[c]{@{}c@{}}Memory (GB)\\(Baseline/\X)\end{tabular}\\
\midrule
zkGPT  & 109.2
& 147.7 (1.00$\times$) & 0.328 (1.00$\times$) & 240.5 (1.00$\times$) \\
\X ($K=2$)  & 141.7
&  96.4 (1.53$\times$) & 0.295 (1.11$\times$) & 257.8 (0.93$\times$) \\
\X ($K=3$)  & 149.4
&  73.1 (2.02$\times$) & 0.287 (1.15$\times$) & 208.3 (1.15$\times$) \\
\X ($K=4$)  & 173.9
&  63.0 (2.34$\times$) & 0.366 (0.90$\times$) & 266.4 (0.90$\times$) \\
\X ($K=6$)  & 200.6
&  45.5 (3.25$\times$) & 0.343 (0.96$\times$) & 208.4 (1.15$\times$) \\
\X ($K=12$) & 268.6
&  \textbf{30.6} (\textbf{4.83}$\times$) & 0.273 (1.20$\times$) & 211.1 (1.14$\times$) \\
\bottomrule
\end{tabular}%
}
\caption{Performance of \X on GPT-2~\cite{gpt} relative to zkGPT ($K=1$) for a sequence length of 64. 
\\ \textsuperscript{\textdagger}Response time $\approx$ Prover time in all cases.}
\label{tab:gpt_clubbed}
\end{table}

\noindent
Table~\ref{tab:gpt_clubbed} presents the performance of \X when partitioning along layer boundaries on
GPT-2~\cite{gpt} in comparison to the baseline (zkGPT~\cite{zkgpt}). Here, we use a sequence length
of 64. We highlight three key observations. 

First, \X achieves a prover-time speedup of up to $4.83\times$ with $12$ partitions, with speedup increasing as the number of partitions grows. For example, with $2$ partitions, \X achieves a $1.53\times$ speedup. This improvement is because independent sub-proofs are executed in parallel, allowing proofs for multiple transformer blocks to be computed simultaneously. As a result, \X reduces total prover time despite the additional cost of committing to boundary polynomials. This additional computation required by \X is $\sim$1--3\% of the total prover time. 

Second, verifier time and peak memory are roughly constant irrespective of the number of partitions. Verifier time remains $\sim0.3s$, although each sub-proof only proves a subset of the computations. However, each sub-proof also includes additional commitment openings for the boundary polynomials. Since the commitment openings dominate the verifier execution time, we do not see a reduction in verifier times. Peak memory ranges from
$0.90\times$ for $2$ partitions to $1.2\times$ for $12$ partitions, which is comparable to the baseline. This is because \X does not modify the total number of operations required to compute the proof significantly. Thus, the intermediate values required to store the witness values and compute the proofs for all partitions are similar in number to the baseline.

Third, the proof size increases with the number of partitions, from
$109.2$KB in the baseline to $268.6$KB for 12 partitions (a $2.5\times$
increase). This is due to the additional boundary commitments and their
openings. This increase, however, does not affect response time: transmitting the
largest proof takes under $3$ms on a 100\,MBps network, so the response
time speedup matches the prover-time speedup of $4.83\times$.
We conclude that the sub-proof parallelism exposed by \X reduces proof-generation latency on GPT-2 workloads, despite the increase in proof size.

\begin{table}[h]
\centering\small
\setlength{\tabcolsep}{6pt}
\begin{tabular}{c c c c}
\toprule
\begin{tabular}[c]{@{}c@{}}Total\\threads ($T$)\end{tabular}
& \begin{tabular}[c]{@{}c@{}}Baseline ($K{=}1$)\\Prover (s)\end{tabular}
& \begin{tabular}[c]{@{}c@{}}\X \\ ($K{=}12$)\\Prover (s)\end{tabular}
& Speedup \\
\midrule
16  & 160.5 & 55.9 & 2.87$\times$          \\
32  & 146.9 & 41.7 & 3.52$\times$          \\
64  & 141.1 & 33.4 & 4.22$\times$          \\
128 & 142.9 & 30.6 & 4.66$\times$          \\
192 & 147.7 & 30.6 & \textbf{4.83}$\times$ \\
\bottomrule
\end{tabular}
\caption{Performance of \X on GPT-2 with 12 partitions for various thread counts.}
\label{tab:gpt_threads}
\end{table}
\noindent
\subsubsection{Scalability of \X with increasing number of cores}
Table~\ref{tab:gpt_threads} presents the performance of \X when partitioning along the model layers to obtain 12 partitions on
GPT-2~\cite{gpt} for a sequence length
of 64 for different threads $T$. Here, the baseline computes its single
proof using all $T$ threads, while the $12$ sub-proofs of \X are executed in
parallel using a total of $T$ threads, \ie each sub-proof uses
$\lfloor T/12\rfloor$ threads. We observe that both the baseline and \X benefit from increasing the total number of threads. The prover time for the baseline reduces from $160.5$s to $147.7$s, while for \X it 
reduces from $55.9$s to $30.6$s, when using 16 and 192 threads, respectively. However, the speedup due to \X also increases from $2.87\times$ when using 16 threads to $4.83\times$ when using 192 threads. This is because the baseline exploits parallelism only within the cryptographic kernels for a single proof. 
In contrast, \X exposes an additional dimension of parallelism, across neural network layers in different partitions. 
Consequently, as the total number of threads increases, more threads can be assigned to each sub-proof, thereby reducing the prover time of each partition. Thus \X yields progressively larger performance gains relative to the baseline with an increase in number of threads.

\begin{table}[t]
\centering\small
\setlength{\tabcolsep}{5pt}
\resizebox{\columnwidth}{!}{%
\begin{tabular}{c c c c c}
\toprule
\begin{tabular}[c]{@{}c@{}}Seq.\\len.\end{tabular} &
\begin{tabular}[c]{@{}c@{}}Partitioning\\$K_\ell \times K_s$ ($K$)\end{tabular} &
\begin{tabular}[c]{@{}c@{}}Proof\\(KB)\end{tabular} &
\begin{tabular}[c]{@{}c@{}}Prover Time (s)\\(Speedup)\end{tabular} &
\begin{tabular}[c]{@{}c@{}}Verifier Time (s)\\(Speedup)\end{tabular} \\
\midrule
\multirow{2}{*}{64}
  & $12 \times 1$ (12) & 268.6  & \textbf{30.6} (\textbf{4.83}$\times$) & 0.273 (1.20$\times$) \\
  & $12 \times 2$ (24) & 956.9  & 43.3 (3.41$\times$)          & 1.268 (0.26$\times$) \\
\cmidrule(l){1-5}
\multirow{3}{*}{128}
  & $12 \times 1$ (12) & 325.6  & 57.0 (3.94$\times$)          & 0.287 (1.19$\times$) \\
  & $12 \times 2$ (24) & 1009.7 & \textbf{55.3} (\textbf{4.06}$\times$) & 1.295 (0.26$\times$) \\
  & $12 \times 4$ (48) & 1937.2 & 63.6 (3.53$\times$)          & 1.423 (0.24$\times$) \\
\cmidrule(l){1-5}
\multirow{3}{*}{256}
  & $12 \times 1$ (12) & 335.8  & 100.4 (6.03$\times$)         & 0.350 (1.13$\times$) \\
  & $12 \times 4$ (48) & 2059.1 & \textbf{88.5} (\textbf{6.84}$\times$) & 2.888 (0.14$\times$) \\
  & $12 \times 8$ (96) & 3949.6 & 110.9 (5.45$\times$)         & 2.432 (0.16$\times$) \\
\bottomrule
\end{tabular}%
}
\caption{Performance of \X on GPT-2 when partitioning along both model-layer and sequence dimensions. $K_\ell$ denotes the number of layer partitions, and $K_s$ denotes the number of sequence partitions.}
\label{tab:gpt_sequence_partition}
\end{table}

\subsubsection{Performance of \X with sequence partitioning}
Table~\ref{tab:gpt_sequence_partition} evaluates \X when we partition GPT-2 along both the model-layer dimension and the sequence dimension. 

We highlight three key observations. First, sequence partitioning provides an additional source of parallelism, thereby improving performance. For example, for sequence length $128$, increasing from $12$ layer-only partitions to $24$ layer-and-sequence partitions reduces prover time from $57.0$s to $55.3$s, improving speedup from $3.94\times$ to $4.06\times$. For sequence length $256$, using $48$ total partitions provides a speedup of $6.84\times$. This is because sequence partitioning exposes parallelism within each transformer block in addition to the parallelism across blocks. As the sequence length grows, each block contains more work, so splitting the sequence lets \X better utilize the available hardware.

Second, sequence partitioning has diminishing returns when the partitions become too fine-grained. For example, with sequence length $64$, increasing from $12$ to $24$ total partitions increases prover time from $30.6$s to $43.3$s. This is because each additional sequence partition introduces more boundary commitments, while also reducing the number of threads available to each sub-proof. Thus, sequence partitioning is most beneficial when the additional parallelism outweighs the cost of computing more sub-proofs and the additional boundary commitments.

Third, proof size and verifier time increase with two-dimensional partitioning. For example, at sequence length $256$, using $12$ layer-only partitions vs $48$ layer-and-sequence partitions increases proof size from $335.8$KB to $2059.1$KB, and verifier time from $0.350$s to $2.888$s. This is because the verifier must check more sub-proofs and more boundary openings across both the layer and sequence dimensions. However, the verifier time remains small relative to prover time. Similarly, the number of additional commitments and openings required for \X increases with increase in number of partitions, increasing the size of the proofs. However, the largest proof in Table~\ref{tab:gpt_sequence_partition} is under $4$MB, which takes under $40$ms to transmit on a 100\,MBps network. Thus, the increase in proof size has a negligible impact on the prover time speedups due to \X.

We conclude that in \X, sequence partitioning complements layer partitioning by increasing the number of available independent sub-proofs.

\begin{table}[t]
  \centering
  \begin{tabular}{cccc}
    \toprule
    \begin{tabular}[c]{@{}c@{}}Seq.\\ len.\end{tabular} &
    \begin{tabular}[c]{@{}c@{}}Baseline peak \\ memory (GB)\end{tabular} &
    \begin{tabular}[c]{@{}c@{}}\X Peak memory \\ (GB) (Baseline/\X)\end{tabular} &
    \begin{tabular}[c]{@{}c@{}}Prover time (s)\\ (Slowdown)\end{tabular} \\
    \midrule
    64  & 240.5 & 31.9 ($7.5\times$) & 288.5 ($1.95\times$) \\
    128 & 269.1 & 52.5 ($5.1\times$) & 479.9 ($2.14\times$) \\
    256 & 532.4 & 66.0 ($8.1\times$) & 846.6 ($1.40\times$) \\
    \bottomrule
  \end{tabular}
    \caption{Peak prover memory and prover time slowdown for zkGPT~\cite{zkgpt} versus sequential \X with
  $12$ partitions.}
  \label{tab:seqmem}
\end{table}

\subsubsection{Sequential Execution and Memory Footprint}
\label{app:seqmem}

\X supports a \emph{sequential} execution mode, where sub-proofs are
generated one at a time rather than concurrently. Table~\ref{tab:seqmem}
presents the performance of sequential \X with $12$ layer partitions on
GPT-2~\cite{gpt} relative to the baseline (zkGPT~\cite{zkgpt}). We make
two observations.

First, sequential \X substantially reduces peak prover memory across all
sequence lengths. For example, at sequence length $64$, peak memory
decreases from $240.5$ GB to $31.9$ GB with $K{=}12$, a $7.5\times$
reduction. \X achieves similar reductions of $5.1\times$ and $8.1\times$
for sequence lengths $128$ and $256$, respectively. This is because each partition covers only a subset of the computations, and only one
partition's prover, witness, and intermediate values are stored in main
memory at a time.

Second, this memory reduction comes at the cost of increased prover
time. For example, sequential execution is $1.95\times$ slower than the
baseline at sequence length $64$. This overhead has two causes.
(i)~\X adds commitments and openings for the boundary polynomial.
(ii) The monolithic baseline merges the arithmetic circuits of all
partitions into a single circuit. This merged circuit allows
circuit-level optimizations across partitions, yielding fewer
constraints at the cost of higher memory.

Despite this time overhead, the memory reduction enables proof
generation on memory-constrained accelerators such as GPUs. Each
sub-proof fits within device memory ($32\text{--}66$\,GB), while the
baseline ($240\text{--}532$\,GB) does not. A host with large main
memory can support the baseline by paging from device memory to main
memory, but such paging is costly. By fitting each sub-proof within
device memory, \X can exploit the GPU's high memory bandwidth. This would substantially reduce per-sub-proof prover time and improve performance.

\section{Limitations}
\X's parallelism scales with the number of partitions $K$. In our implementation, $K$ is determined by the model architecture and the input size. Along the layer dimension, $K$ is set by the model's natural boundary points (\eg 5 for AlexNet). Models such as LLMs expose an additional dimension along the input sequence. \X therefore provides the greatest benefit for models and inputs that admit many partitions. We leave the design of additional partitioning schemes that further increase $K$ to future work.

%% file: src/prior.tex
\section{Related Work}
Here, we discuss prior work on improving prover performance and on modular proof frameworks.

\noindent
\textbf{Efficient zkML implementations and Arithmetization.}
Many prior works~\cite{298276, zkcnn,zkml,zkllm,zkgpt,cryptoeprint:2024/162} improve zkML performance by reducing constraint counts and optimizing arithmetization for common ML operators. zkCNN~\cite{zkcnn}, zkGPT~\cite{zkgpt}, and others~\cite{zkllm,qu_verfcnn_2025} introduce specialized arithmetizations for ML operators that reduce proof-generation cost, \eg convolution operations in zkCNN. Compiler and system-oriented frameworks~\cite{zktorch,zkpytorch,zkml,feng2024zeno} such as ZEN~\cite{feng2021zen} exploit high-level tensor semantics to generate specialized circuits for linear-algebra kernels. Artemis~\cite{apollo} proposes a CP-SNARK construction that lets the prover additionally prove the model weights match an external commitment. These works improve performance for a single proof. Our construction is complementary: it can be combined with them to improve prover performance via parallel sub-proof generation. We demonstrate this by implementing \X on top of zkGPT~\cite{zkgpt} and zkCNN~\cite{zkcnn} without modifying their arithmetization or PCS. Each partition's proof thus retains the original optimizations proposed by these works.

\noindent
\textbf{Modularity and proof composition.}
A complementary line of work constructs proofs from smaller proof components. LegoSNARK~\cite{campanelli2019legosnark} formalizes a modular approach for designing and composing SNARKs, assembling proofs of larger relations from proofs of simpler sub-relations. Commit-and-prove constructions such as Lunar~\cite{campanelli2021lunar} provide cryptographic mechanisms to enforce relationships (\eg equality) between a proof's witnesses and an externally committed witness.
In principle, these constructions can combine arguments from independent proof systems. However, Lunar and LegoSNARK cannot be used directly with zkML frameworks (see \S\ref{sec:key_idea}). Further, Lunar requires \prover to compute a linking proof for equality between witness values in two commitments. In contrast, \X's shared commitments require no linking proof, improving performance.

Prior work~\cite{ml_modular} introduces a modular framework based on verifiable evaluation (VE) for ML applications. It flexibly combines proof gadgets (protocols) to efficiently prove complex ML computations. Our construction is orthogonal: in the VE framework, proof generation remains sequential, from the outer layer to the inner layer. Integrating it with \X improves prover time per sub-proof, yielding better overall performance.

\noindent
\textbf{Accelerating prover performance using hardware.}
Prior work also targets performance by accelerating the dominant cryptographic kernels used by SNARK provers. GPU-based frameworks~\cite{ma2023gzkp,li2025zkpog,10.1145/3774934.3786448} offload expensive operations (\eg MSM/NTT) to GPUs, and libraries such as Icicle~\cite{kthiri2024icicle} provide optimized GPU kernels for integration into SNARK stacks. Several works~\cite{samardzic2024nocap,daftardar2024szkp,wang2025unizk,daftardar2025zkspeed} also propose dedicated hardware accelerators for proof generation. These approaches are complementary to \X. \X unlocks an additional dimension of parallelism, \ie enabling parallel proof generation across partitions. It can be combined with kernel-level GPU optimizations or specialized accelerators to further improve zkML performance.

\noindent
\textbf{Accelerating GKR and sumcheck protocols.}
Prior works such as Libra~\cite{libra}, divide-and-conquer sumcheck~\cite{sumcheck}, SVO~\cite{cryptoeprint:2025/1117}, and Virgo~\cite{virgo} improve GKR and sumcheck protocols by reducing the asymptotic complexity of the proving algorithm. These approaches are orthogonal to ours: they can be applied independently to each sub-proof, yielding better overall performance when combined with \X.


Hydra~\cite{hydra} parallelizes GKR-style proofs by decomposing them across GKR layers, with sub-proofs linked via polynomial commitments. It does not guarantee zero knowledge for the resulting SNARKs, and therefore cannot be used directly for zkML. Modern zkML frameworks such as zkGPT~\cite{zkgpt} use circuit squashing, which reduces circuit depth. Squashing, however, produces layers of non-uniform width. Partitioning across these circuit layers yields sub-proofs with unequal prover times and limited parallel speedup. \X instead partitions at the model level and applies circuit-level optimizations within each partition, yielding balanced sub-proofs.

Some prior works~\cite{full_account,vu_hybrid_2013,cormode_practical_2012,10.1145/3658644.3690318, blendysc} explore techniques for reducing the prover time or memory of the sumcheck protocol in the GKR framework for data-parallel circuits, \ie multiple identical computations on independent data. These optimizations can improve existing zkML systems modeled as data-parallel circuits. \X is orthogonal to these prior approaches. The same techniques can be applied in \X for each sub-proof.

%% file: src/conclusion.tex
\section{Conclusion}
We introduce \X, the first modular proof framework for zkML
applications. \X partitions the proof computation of an ML model into
multiple independent sub-proofs, each proving the correctness of a
subset of the inference operations. This design enables parallel computation of sub-proofs, reducing prover time. Each sub-proof also requires less memory, thus, enabling proof
computation on memory-constrained accelerators. \X requires no changes
to existing polynomial commitment schemes and delivers significant
speedups over state-of-the-art zkML implementations.

%% file: src/appendix.tex
\section{Background Definitions}
\subsection{Zero-Knowledge Sumcheck} 
\label{sec:zk_sumcheck}
To prevent \verifier from recovering the layer polynomial evaluations during sumcheck, \prover modifies the original polynomial by adding a random masking polynomial $k$~\cite{zksumcheck,libra}. \prover samples a random univariate polynomial $k \colon \mathbb{F}^\ell \to \mathbb{F}$ of degree $d = \deg(h)$. The sumcheck protocol is then run on the modified equation~\cite{libra}
\begin{align}
\label{eq:zk-sumcheck}
    H + \rho K = \sum_{x \in \{0,1\}^m} \bigl( h(x) + \rho \, k(x) \bigr),
\end{align}
where $H = \sum_{x \in \{0,1\}^m} h(x)$ is the original claim, $K = \sum_{x \in \{0,1\}^m} k(x)$ is a single scalar sent by \prover, and $\rho \in \mathbb{F}$ is a random challenge chosen by \verifier. Since each evaluation sent during sumcheck is masked by a fresh evaluation of $k$, \verifier learns \emph{nothing} about the individual values of $h$, \ie $\tilde{V}_{i}$ in GKR.

\subsection{ZK-SNARK properties} 
\label{sec:zk-snark}
\begin{definition}[zk-SNARK properties]
\label{def:zksnark}
A zk-SNARK for relation $\mathcal{R}$ satisfies:
\begin{itemize}
    \item \textbf{Completeness:} For all $(\mathsf{i}, \mathsf{x}, \mathsf{w}) \in \mathcal{R}$:
    \begin{align*}
    \Pr\Big[\mathsf{Verify}(\mathsf{pp}, \mathsf{i}, \mathsf{x}, \pi) = 1\Big] = 1
    \end{align*}
    where $\mathsf{pp} \leftarrow \mathsf{Setup}(1^\lambda, \mathsf{i})$ and $\pi \leftarrow \mathsf{Prove}(\mathsf{pp}, \mathsf{i}, \mathsf{x}, \mathsf{w})$.
    
    \item \textbf{Knowledge Soundness:} For every PPT adversary $\mathcal{A}$, there exists a PPT extractor $\mathcal{E}$ such that for all $(\mathsf{i}, \mathsf{x})$:
    \begin{align*}
    \Pr\Big[\mathsf{Verify}(\mathsf{pp}, \mathsf{i}, \mathsf{x}, \pi) = 1 \;\land\; (\mathsf{i}, \mathsf{x}, \mathsf{w}) \notin \mathcal{R}\Big] \leq \mathsf{negl}(\lambda)
    \end{align*}
    where $\mathsf{pp} \leftarrow \mathsf{Setup}(1^\lambda, \mathsf{i})$, $(\mathsf{x}, \pi) \leftarrow \mathcal{A}(\mathsf{pp}, \mathsf{i})$, and $\mathsf{w} \leftarrow \mathcal{E}^{\mathcal{A}}(\mathsf{pp}, \mathsf{i})$.
    
    \item \textbf{Zero-Knowledge:} There exists a PPT simulator $\mathsf{Sim}$ such that for all $(\mathsf{i}, \mathsf{x}, \mathsf{w}) \in \mathcal{R}$ and all PPT distinguishers $\mathcal{D}$:
    \begin{align*}
    \big| \Pr[\mathcal{D}(\pi) = 1] - \Pr[\mathcal{D}(\pi') = 1] \big| \leq \mathsf{negl}(\lambda)
    \end{align*}
    where $\pi \leftarrow \mathsf{Prove}(\mathsf{pp}, \mathsf{i}, \mathsf{x}, \mathsf{w})$ and $\pi' \leftarrow \mathsf{Sim}(\mathsf{pp}, \mathsf{i}, \mathsf{x})$.
\end{itemize}
\end{definition}

\subsection{Polynomial Commitment Scheme}
\label{sec:pcs_formal}
Let \(\mathcal{P}_D\) denote the class of polynomials over \(\mathbb{F}_p\) with degree bounds \(D\). A polynomial commitment scheme (PCS) for \(\mathcal{P}_{D}\) consists of the following algorithms.
\begin{itemize}
    \item \(\mathsf{Setup}(1^\lambda,D)\rightarrow \mathsf{pp}\):
On input security parameter \(\lambda\) and class description \(D\), the setup algorithm outputs public parameters \(\mathsf{pp}\).

    \item \(\mathsf{Commit}(\mathsf{pp}, f; r) \rightarrow \mathsf{com}\):
    On input \(\mathsf{pp}\), polynomial \(f \in \mathcal{P}_{s,D}\), and commitment randomness \(r\), the commitment algorithm outputs commitment \(\mathsf{com}\).

    \item \(\mathsf{Open}(\mathsf{pp},f,\boldsymbol{\alpha};r)\rightarrow(v,\pi)\):
On input \(f\in\mathcal{P}_D\), evaluation point \(\boldsymbol{\alpha}\in\mathbb{F}_p^{s_f}\), and commitment randomness \(r\), the opening algorithm outputs \(v=f(\boldsymbol{\alpha})\) and opening proof \(\pi\).

    \item \(\mathsf{Verify}(\mathsf{pp},\mathsf{com},\boldsymbol{\alpha},v,\pi)
    \rightarrow \{0,1\}\):
    On input \(\mathsf{pp}\), commitment \(\mathsf{com}\), evaluation point \(\boldsymbol{\alpha}\in\mathbb{F}_p^{s_f}\), claimed value \(v\), and opening proof \(\pi\), the verifier outputs \(1\) if the opening is accepted and \(0\) otherwise.
\end{itemize}

\begin{definition}[PCS properties]
\label{def:pcs}
A polynomial commitment scheme $(\mathsf{Setup},\mathsf{Commit},\mathsf{Open},\mathsf{Verify})$
for the polynomial class \(\mathcal{P}_D\) satisfies the following
properties.

\smallskip
\noindent\textbf{Completeness.}
For every polynomial \(f\in\mathcal{P}_D\), every evaluation point
\(\alpha\) in \(f\)'s domain, and every valid commitment randomness \(r\),
we have
\[
\Pr\left[
\mathsf{Verify}(\mathsf{pp},\mathsf{com},\alpha,v,\pi)=1
\right]=1,
\]
where
\[
\mathsf{pp}\leftarrow \mathsf{Setup}(1^\lambda,D),\qquad
\mathsf{com}\leftarrow \mathsf{Commit}(\mathsf{pp},f;r),
\]
and
\[
(v,\pi)\leftarrow \mathsf{Open}(\mathsf{pp},f,\alpha;r).
\]

\smallskip
\noindent\textbf{Binding.}
For every \(\PPT\) adversary \(\mathcal{A}\),
\[
\Pr\left[
\begin{array}{l}
\mathsf{Verify}(\mathsf{pp},\mathsf{com},\alpha,v,\pi)=1
\;\land \\[1mm]
\mathsf{Verify}(\mathsf{pp},\mathsf{com},\alpha,v',\pi')=1
\;\land \\[1mm]
v\neq v'
\end{array}
\right]\leq \negl(\lambda),
\]
where
$\mathsf{pp}\leftarrow\mathsf{Setup}(1^\lambda,D),$
and
$
(\mathsf{com},\alpha,v,\pi,v',\pi')
\leftarrow \mathcal{A}(\mathsf{pp}).
$

\smallskip
\noindent\textbf{Hiding.}
For security parameter \(\lambda\), public parameters
\(\mathsf{pp}\leftarrow\mathsf{Setup}(1^\lambda,D)\), polynomial
\(f\in\mathcal{P}_D\), \(\PPT\) adversary \(\mathcal{A}\), and simulator
\(\mathsf{S}=(\mathsf{S}_1,\mathsf{S}_2)\), consider the following two
experiments.

\[
\begin{array}[t]{l}
\underline{\mathsf{Real}_{\mathcal{A},f}(\mathsf{pp}):}\\[1mm]
r_f \leftarrow \mathcal{R};\\
\mathsf{com}\leftarrow\mathsf{Commit}(\mathsf{pp},f;r_f);\\
\mathsf{st}\leftarrow\mathcal{A}(\mathsf{com});\\
(\alpha,y,\pi)\leftarrow \\
\quad\langle \mathsf{Open}(\mathsf{pp},f,\cdot\,;r_f),\mathcal{A}\rangle(\mathsf{st});\\
b\leftarrow\mathcal{A}(\mathsf{com},\alpha,y,\pi);\\
\text{output } b.
\end{array}
\qquad
\begin{array}[t]{l}
\underline{\mathsf{Ideal}_{\mathcal{A},f}^{\mathsf{S}}(\mathsf{pp}):}\\[1mm]
(\mathsf{com},\mathsf{td})\leftarrow\mathsf{S}_1(\mathsf{pp});\\
\mathsf{st}\leftarrow\mathcal{A}(\mathsf{com});\\
\alpha\leftarrow\mathcal{A}(\mathsf{st});\\
y= f(\alpha);\\
\pi\leftarrow\mathsf{S}_2(\mathsf{td},\alpha,y);\\
b\leftarrow\mathcal{A}(\mathsf{com},\alpha,y,\pi);\\
\text{output } b.
\end{array}
\]

The PCS is hiding if, for every \(\PPT\) adversary \(\mathcal{A}\) and every
\(f\in\mathcal{P}_D\), there exists a \(\PPT\) simulator
$\mathsf{S}$ such that
\[
\left|
\Pr\left[\mathsf{Real}_{\mathcal{A},f}(\mathsf{pp})=1\right]
-
\Pr\left[\mathsf{Ideal}_{\mathcal{A},f}^{\mathsf{S}}(\mathsf{pp})=1\right]
\right|
\leq \negl(\lambda).
\]
\end{definition}

\section{Discussion}
\subsection{Multiple Parts Consuming The Same Input}
\label{sec:mult_part}
In models with residual or skip connections (\eg ResNet~\cite{resnet}), a layer's output may be consumed by multiple non-contiguous downstream layers. If such a layer is chosen as a partition boundary, \X can reuse a single commitment for the shared activations across multiple sub-proofs. However, the zero-knowledge construction requires care: when a committed polynomial is opened in more than two sub-proofs, the masking polynomial ($R$) used in its low-degree extension must include sufficiently many monomials to ensure that all revealed evaluations remain statistically independent. The required number of variables and monomials in $R$ therefore grows with the number of consumer proofs. For example, if the commitment is shared across 4 to 9 proofs, $R$ can be defined as:
\begin{align}
\label{eq:masking-3}
    R(z_2,z_1,w) = & a_0 + a_1z_2+a_2z_2^2+a_3z_1+a_4z_1^2+a_5w^2+ \nonumber\\ &a_6w+a_7z_2z_1+a_8z_2w+a_9z_2^2z_1^2w^2
\end{align}
where $a_0,\dots,a_9$ are randomly sampled from $\mathbb{F}_p$.

\subsection{Applying \X To Other Computations}
\label{sec:other_computation}
The \X construction only modifies the GKR protocol at the boundary layers, leaving the core protocol unchanged. Hence, \X is not specific to machine learning workloads. \X can be applied to any computation that can be represented as a layered arithmetic circuit verified by sumcheck-based protocols. 

\subsection{Applying \X To PLONK-based Systems}
\label{sec:plonk}
Although we present \X for GKR-based zkML systems, the underlying idea of reusing committed intermediate values across independently verified sub-proofs also applies to table-based arithmetizations like PLONK~\cite{Gabizon2019PLONKPO}. In such settings, \X can be realized by assigning shared boundary values to fixed advice columns and reusing them consistently across circuits, similar to techniques in Artemis~\cite{apollo}. The main requirement is that the commitment scheme enforces binding across all sub-proofs that reference the shared values.

\subsection{\X: Analysis}
\label{sec:proof}

\noindent
\begin{proof}[Proof of Theorem~\ref{thm:knowledge-argument-twopart}]
We prove completeness and knowledge soundness separately.

\smallskip
\noindent
\textbf{Completeness.}
Let \((\mathsf{i},\mathsf{x},\mathsf{w})\in\mathcal{R}_{\mathsf{zkml}}\),
where \(\mathsf{x}=\{q,r\}\). Consider an honest execution of \(\X\) with
witness \(\mathsf{w}\). By completeness of the PCS, every honestly generated opening of
\(\mathsf{com}_b\), \(\mathsf{com}_{R_b}\), and the other commitments used in
Protocol~\ref{prot:zkmod} verifies. By completeness of the underlying
GKR-based CP-SNARKs, the sub-proofs for \(\mathcal{R}_A\) and
\(\mathcal{R}_B\) both accept. Since the \(\X\) verifier accepts only if both
sub-proofs accept and all PCS openings verify, the honestly generated proof
\(\pi\leftarrow\Prove(\mathsf{pp},\mathsf{i},\mathsf{x},\mathsf{w})\) is
accepted. Therefore, \(\X\) satisfies completeness for
\(\mathcal{R}_{\mathsf{zkml}}\).

\smallskip
\noindent
\textbf{Knowledge Soundness.}
Let \(\mathcal{A}\) be any \(\PPT\) adversary that, on input public parameters
\(\mathsf{pp}\), index \(\mathsf{i}\), and instance
\(\mathsf{x}=\{q,r\}\), outputs an accepting \(\X\) proof
\[
\pi=(\mathsf{com}_b,\pi_A,\pi_B,\ldots)
\]
with non-negligible probability. We construct a \(\PPT\) extractor that outputs
a valid witness for \(\mathcal{R}_{\mathsf{zkml}}\).

Since the \(\X\) verifier accepts, both underlying sub-verifiers accept. By
knowledge soundness of the CP-SNARK for \(\mathcal{R}_A\), an extractor returns
a witness
\[
\mathsf{w}_A=
(\theta_A,\mathsf{adv}_A,\mathsf{w}_b)
\]
such that $(\mathsf{i}_A,\mathsf{x}_A=\{r\},\mathsf{w}_A)\in\mathcal{R}_A .$
Similarly, the CP-SNARK for \(\mathcal{R}_B\) yields an extractor returning
a witness
\[
\mathsf{w}_B=
(\theta_B,\mathsf{adv}_B,\mathsf{act}_b)
\]
such that $(\mathsf{i}_B,\mathsf{x}_B=\{q\},\mathsf{w}_B)\in\mathcal{R}_B .$

It remains to show that these extracted boundary values are consistent. In
Protocol~\ref{prot:zkmod}, both sub-proofs use the same PCS commitment $\mathsf{com}_b$ to the masked boundary polynomial \(\dot V_b\). The
\(\mathcal{R}_A\) sub-proof opens it when checking the boundary input
polynomial for part~\(A\). The \(\mathcal{R}_B\) sub-proof opens it as the
committed output polynomial for layer \(b\). Since the \(\X\) verifier accepts
only if all such openings verify and the PCS is binding, all accepting openings
of \(\mathsf{com}_b\) are consistent with a single degree-bounded polynomial,
except with negligible probability. By construction (Eq.~\ref{eq:lde}),
\(\dot V_b\) and \(\tilde V_b\) agree on the Boolean hypercube. Therefore, except with negligible probability, the
boundary values \(\mathsf{w}_b\) extracted from the \(\mathcal{R}_A\) proof
equal the layer-\(b\) activations \(\mathsf{act}_b\) extracted from the
\(\mathcal{R}_B\) proof. We define the combined witness
\[
\mathsf{w}=
(\theta_A,\theta_B,\mathsf{adv}_A,\mathsf{adv}_B).
\]
The witness \(\mathsf{w}_A\) enforces correctness of layers
\(\{0,\ldots,b-1\}\) and \(\mathsf{w}_B\) of layers \(\{b,\ldots,d\}\). Since
the extracted boundary values match, these two witnesses combine into a valid
witness for the original end-to-end relation:
$
(\mathsf{i},\mathsf{x},\mathsf{w})\in\mathcal{R}_{\mathsf{zkml}},
$
except with negligible probability.
Therefore, \(\X\) satisfies knowledge soundness for \(\mathcal{R}_{\mathsf{zkml}}\).

Since \(\X\) satisfies both completeness and knowledge soundness, it is a
knowledge argument for \(\mathcal{R}_{\mathsf{zkml}}\).
\end{proof}

\begin{proof}[Proof of Theorem~\ref{thm:zk}]
We construct a simulator \(\mathcal{S}_{\mathsf{mod}}\) for the two-partition
protocol \(\X\). It combines \(\mathcal{S}_{\mathsf{sc}}\) for the underlying
GKR-based CP-SNARK transcripts with \(\mathcal{S}_{\mathsf{pc}}\) for the
hiding PCS.

All transcript components that do not involve the shared boundary polynomial \(\dot V_b\) are simulated as in the standard commit-and-prove GKR analysis. By the zero-knowledge of the underlying GKR-based CP-SNARKs and the hiding of the PCS, these components are indistinguishable from those in a real execution. It remains to consider the messages involving \(\dot V_b\). The verifier learns two evaluations of \(\dot V_b\): one at the input-layer query point \(g^{(\mathsf{inp},A)}\) in sub-protocol \(A\), and one at the layer-\(b\) query point \(g^{(b)}\) in sub-protocol \(B\). In addition, the zero-knowledge GKR oracle check for layer \(b\) reveals one evaluation of the boundary masking polynomial, \(R_b(g^{(b)}_1,c)\), for a verifier-chosen random \(c\in\mathbb{F}_p\).
For simplicity, let \[ u = g^{(b)} \qquad\text{and}\qquad v = g^{(\mathsf{inp},A)}. \] Using the LDE form in Eq.~\ref{eq:lde}, the revealed boundary values can be written as \[ \dot V_b(u) = \widetilde V_b(u) + Z_b(u)M_u, \qquad \dot V_b(v) = \widetilde V_b(v) + Z_b(v)M_v, \] where \[ M_u= \sum_{w\in\{0,1\}}R_b(u_1,w), \qquad M_v= \sum_{w\in\{0,1\}}R_b(v_1,w), \] and the third revealed value is \[ \rho = R_b(u_1,c). \] We analyze the distribution of the mask evaluations \((M_u,M_v,\rho)\). By the definition of \(R_b\) in Eq.~\ref{eq:masking-2}, these values are linear in the random coefficients \(a_0,\ldots,a_6\): \[ \begin{pmatrix} M_u\\ M_v\\ \rho \end{pmatrix} = \begin{pmatrix} 2 & 2u_1 & 1 & u_1 & 2u_1^2 & 1 & u_1^2 \\ 2 & 2v_1 & 1 & v_1 & 2v_1^2 & 1 & v_1^2 \\ 1 & u_1 & c & cu_1 & u_1^2 & c^2 & c^2u_1^2 \end{pmatrix} \begin{pmatrix} a_0\\ a_1\\ a_2\\ a_3\\ a_4\\ a_5\\ a_6 \end{pmatrix}. \] Row reduction shows that the matrix has rank \(3\) whenever \[ u_1\neq v_1,\qquad 2c-1\neq 0,\qquad 2c^2-1\neq 0 \pmod p. \] These conditions are required for zero-knowledge of the underlying SNARKs~\cite{libra}. If any condition fails, the challenges are resampled~\cite{libra}. Thus, the linear map from \((a_0,\ldots,a_6)\) to \((M_u,M_v,\rho)\) is surjective. Since \(a_0,\ldots,a_6\) are sampled uniformly and independently, $(M_u,M_v,\rho)$ is uniform over \(\mathbb{F}_p^3\). Therefore, the revealed boundary-related values are statistically independent of \(\widetilde V_b\).

The simulator \(\mathcal{S}_{\mathsf{mod}}\) samples $Y_u,Y_v,\rho' \xleftarrow{\$}\mathbb{F}_p $ in place of $\dot V_b(u), \dot V_b(v), R_b(u_1,c)$, respectively. It uses \(\mathcal{S}_{\mathsf{pc}}\) to simulate the commitments and openings for \(\mathsf{com}_b\) and \(\mathsf{com}_{R_b}\), and \(\mathcal{S}_{\mathsf{sc}}\) to simulate the GKR/sumcheck transcripts consistently with \(Y_u,Y_v,\rho'\). We compare the real and simulated transcripts by a standard hybrid argument: first replace the GKR/CP-SNARK sub-proofs with \(\mathcal{S}_{\mathsf{sc}}\) outputs, then the PCS commitments and openings with \(\mathcal{S}_{\mathsf{pc}}\) outputs, and finally the boundary-related values with uniform field elements. The first two replacements are computationally indistinguishable by the underlying SNARK's zero-knowledge and the PCS's hiding properties. The last is statistically identical, as shown above. Thus, the real and simulated transcripts are indistinguishable. For every \(\PPT\) distinguisher \(\mathcal{D}\), \[ \left| \Pr[\mathcal{D}(\pi)=1] - \Pr[\mathcal{D}(\pi')=1] \right| \leq \negl(\lambda). \] Therefore, \(\X\) is zero-knowledge for \(\mathcal{R}_{\mathsf{zkml}}\).
\end{proof}

\textbf{Asymptotic Analysis.}
For a log-space uniform circuit $C$ of depth $d$ with input size $n$, the prover complexity is $\mathcal{O}(|C|)$, the verifier complexity is $\mathcal{O}(|x| + d \cdot \log |C|)$, and the proof size is $\mathcal{O}(d \cdot \log |C|)$~\cite{libra}.
Under \X, each sub-proof's complexities---with input sizes $n_1, n_2$ and circuits $C_1, C_2$ of depths $d_1, d_2$---match the baseline GKR protocol. The underlying zk-SNARK protocols for the sub-proofs remain unmodified. Thus, for $k$ partitions, the $i$-th proof has prover complexity $\mathcal{O}(|C_i|)$, verifier complexity $\mathcal{O}(|x_i| + d_i \cdot \log |C_i|)$, and proof size $\mathcal{O}(d_i \cdot \log |C_i|)$. Total complexities follow by summing across partitions. For equal partitions ($|C_i|=|C|/k$, $|x_i|=|x|/k$), the total prover complexity under \X matches the monolithic proof. However, the additional polynomial commitments increase total proof size in practice.

%% file: ref.bib
@article{zhang2019deep,
  author    = {Zhang, Shuai and Yao, Lina and Sun, Aixin and Tay, Yi},
  title     = {Deep Learning Based Recommender System: A Survey and New Perspectives},
  journal   = {ACM Computing Surveys},
  volume    = {52},
  number    = {1},
  pages     = {1--38},
  year      = {2019},
  publisher = {ACM}
}

@inproceedings {298276,
author = {Meng Hao and Hanxiao Chen and Hongwei Li and Chenkai Weng and Yuan Zhang and Haomiao Yang and Tianwei Zhang},
title = {Scalable Zero-knowledge Proofs for Non-linear Functions in Machine Learning},
booktitle = {33rd USENIX Security Symposium (USENIX Security 24)},
year = {2024},
isbn = {978-1-939133-44-1},
address = {Philadelphia, PA},
pages = {3819--3836},
url = {https://www.usenix.org/conference/usenixsecurity24/presentation/hao-meng-scalable},
publisher = {USENIX Association},
month = aug
}

@article{gpt,
  title={Language Models are Unsupervised Multitask Learners},
  author={Radford, Alec and Wu, Jeffrey and Child, Rewon and Luan, David and Amodei, Dario and Sutskever, Ilya},
  journal={OpenAI blog},
  volume={1},
  number={8},
  pages={9},
  year={2019}
}

@inproceedings{blendysc,
author = {Baweja, Anubhav and Chiesa, Alessandro and Fedele, Elisabetta and Fenzi, Giacomo and Mishra, Pratyush and Mopuri, Tushar and Zitek-Estrada, Andrew},
title = {Time-Space Trade-Offs for Sumcheck},
year = {2025},
isbn = {978-3-032-12289-6},
publisher = {Springer-Verlag},
address = {Berlin, Heidelberg},
url = {https://doi.org/10.1007/978-3-032-12290-2_2},
doi = {10.1007/978-3-032-12290-2_2},
booktitle = {Theory of Cryptography: 23rd International Conference, TCC 2025, Aarhus, Denmark, December 1–5, 2025, Proceedings, Part IV},
pages = {37–70},
numpages = {34},
location = {Aarhus, Denmark}
}

@inproceedings{pippenger1976evaluation,
  title={On the Evaluation of Powers and Related Problems},
  author={Pippenger, Nicholas},
  booktitle={17th Annual Symposium on Foundations of Computer Science (SFCS 1976)},
  pages={258--263},
  year={1976},
  organization={IEEE}
}

@misc{vgg16,
      title={Very Deep Convolutional Networks for Large-Scale Image Recognition}, 
      author={Karen Simonyan and Andrew Zisserman},
      year={2015},
      eprint={1409.1556},
      archivePrefix={arXiv},
      primaryClass={cs.CV},
      url={https://arxiv.org/abs/1409.1556}, 
}

@inproceedings{alexnet,
  title={ImageNet Classification with Deep Convolutional Neural Networks},
  author={Krizhevsky, Alex and Sutskever, Ilya and Hinton, Geoffrey E},
  booktitle={Advances in Neural Information Processing Systems},
  volume={25},
  pages={1097--1105},
  year={2012}
}

@inproceedings{chiesa2020marlin,
  title={Marlin: Preprocessing zkSNARKs with Universal and Updatable SRS},
  author={Chiesa, Alessandro and Hu, Yuncong and Maller, Mary and Mishra, Pratyush and Vesely, Noah and Ward, Nicholas},
  booktitle={Advances in Cryptology--EUROCRYPT 2020: 39th Annual International Conference on the Theory and Applications of Cryptographic Techniques, Zagreb, Croatia, May 10--14, 2020, Proceedings, Part I 39},
  pages={738--768},
  year={2020},
  organization={Springer}
}

@article{lee2024vcnn,
  author={Lee, Seunghwan and Ko, Hankyung and Kim, Jihye and Oh, Hyunok},
  title={{vCNN}: Verifiable Convolutional Neural Network Based on zk-{SNARKs}},
  journal={IEEE Transactions on Dependable and Secure Computing},
  year={2024},
  volume={21},
  number={4},
  pages={4254-4270},
  doi={10.1109/TDSC.2023.3348760}
}

@inproceedings{vu_hybrid_2013,
	address = {Berkeley, CA},
	title = {A {Hybrid} {Architecture} for {Interactive} {Verifiable} {Computation}},
	isbn = {9780769549774 9781467361668},
	url = {http://ieeexplore.ieee.org/document/6547112/},
	doi = {10.1109/SP.2013.48},
	urldate = {2026-01-27},
	booktitle = {2013 {IEEE} {Symposium} on {Security} and {Privacy}},
	publisher = {IEEE},
	author = {Vu, V. and Setty, S. and Blumberg, A. J. and Walfish, M.},
	month = may,
	year = {2013},
	pages = {223--237},
}

@misc{ezkl,
	title = {zkonduit/ezkl},
	url = {https://github.com/zkonduit/ezkl},
	urldate = {2026-01-27},
	publisher = {Zkonduit},
	month = jan,
	year = {2026},
	note = {original-date: 2022-07-05T19:54:03Z},
}

@misc{cryptoeprint:2024/162,
      author = {Kasra Abbaszadeh and Christodoulos Pappas and Jonathan Katz and Dimitrios Papadopoulos},
      title = {Zero-Knowledge Proofs of Training for Deep Neural Networks},
      howpublished = {Cryptology {ePrint} Archive, Paper 2024/162},
      year = {2024},
      url = {https://eprint.iacr.org/2024/162}
}

@INPROCEEDINGS {11023367,
author = { Datta, Trisha and Chen, Binyi and Boneh, Dan },
booktitle = { 2025 IEEE Symposium on Security and Privacy (SP) },
title = {{VerITAS: Verifying Image Transformations at Scale}},
year = {2025},
volume = {},
ISSN = {},
pages = {4606-4623},
doi = {10.1109/SP61157.2025.00097},
url = {https://doi.ieeecomputersociety.org/10.1109/SP61157.2025.00097},
publisher = {IEEE Computer Society},
address = {Los Alamitos, CA, USA},
month =May}

@inproceedings{cormode_practical_2012,
	address = {Cambridge Massachusetts},
	title = {Practical verified computation with streaming interactive proofs},
	isbn = {9781450311151},
	url = {https://dl.acm.org/doi/10.1145/2090236.2090245},
	doi = {10.1145/2090236.2090245},
	language = {en},
	urldate = {2026-01-27},
	booktitle = {Proceedings of the 3rd {Innovations} in {Theoretical} {Computer} {Science} {Conference}},
	publisher = {ACM},
	author = {Cormode, Graham and Mitzenmacher, Michael and Thaler, Justin},
	month = jan,
	year = {2012},
	pages = {90--112},
}

@inproceedings{full_account,
author = {Wahby, Riad S. and Ji, Ye and Blumberg, Andrew J. and Shelat, Abhi and Thaler, Justin and Walfish, Michael and Wies, Thomas},
title = {Full Accounting for Verifiable Outsourcing},
year = {2017},
isbn = {9781450349468},
publisher = {Association for Computing Machinery},
address = {New York, NY, USA},
url = {https://doi.org/10.1145/3133956.3133984},
doi = {10.1145/3133956.3133984},
booktitle = {Proceedings of the 2017 ACM SIGSAC Conference on Computer and Communications Security},
pages = {2071–2086},
numpages = {16},
location = {Dallas, Texas, USA},
series = {CCS '17}
}

@inproceedings{10.1145/3774934.3786448,
author = {Zhang, Zhiyuan and Cai, Yanxin and Yin, Wenhao and Wu, Xueyu and Wang, Yi and Ju, Lei and Ji, Zhuoran},
title = {Pipelonk: Accelerating End-to-End Zero-Knowledge Proof Generation on GPUs for PLONK-Based Protocols},
year = {2026},
isbn = {9798400723100},
publisher = {Association for Computing Machinery},
address = {New York, NY, USA},
url = {https://doi.org/10.1145/3774934.3786448},
doi = {10.1145/3774934.3786448},
booktitle = {Proceedings of the 31st ACM SIGPLAN Annual Symposium on Principles and Practice of Parallel Programming},
pages = {439–451},
numpages = {13},
location = {Sydney, NSW, Australia},
series = {PPoPP '26}
}

@inproceedings{10.1145/3658644.3690318,
author = {Pappas, Christodoulos and Papadopoulos, Dimitrios},
title = {Sparrow: Space-Efficient zkSNARK for Data-Parallel Circuits and Applications to Zero-Knowledge Decision Trees},
year = {2024},
isbn = {9798400706363},
publisher = {Association for Computing Machinery},
address = {New York, NY, USA},
url = {https://doi.org/10.1145/3658644.3690318},
doi = {10.1145/3658644.3690318},
booktitle = {Proceedings of the 2024 on ACM SIGSAC Conference on Computer and Communications Security},
pages = {3110–3124},
numpages = {15},
location = {Salt Lake City, UT, USA},
series = {CCS '24}
}

@misc{qu_verfcnn_2025,
	title = {{VerfCNN}, {Optimal} {Complexity} {zkSNARK} for {Convolutional} {Neural} {Networks}},
	url = {https://eprint.iacr.org/2025/2020},
	urldate = {2026-01-27},
	author = {Qu, Wenjie and Guo, Yanpei and Ying, Yue and Zhang, Jiaheng},
	year = {2025},
	note = {Publication info: Published elsewhere. Minor revision. IEEE S\&P 2026},
}

@misc{zkpytorch,
	title = {{zkPyTorch}: {A} {Hierarchical} {Optimized} {Compiler} for {Zero}-{Knowledge} {Machine} {Learning}},
	shorttitle = {{zkPyTorch}},
	url = {https://eprint.iacr.org/2025/535},
	urldate = {2026-01-27},
	author = {Xie, Tiancheng and Lu, Tao and Fang, Zhiyong and Wang, Siqi and Zhang, Zhenfei and Jia, Yongzheng and Song, Dawn and Zhang, Jiaheng},
	year = {2025},
	note = {Publication info: Preprint.},
}

@inproceedings{ml_modular,
author = {Balb\'{a}s, David and Fiore, Dario and Gonz\'{a}lez Vasco, Maria Isabel and Robissout, Damien and Soriente, Claudio},
title = {Modular Sumcheck Proofs with Applications to Machine Learning and Image Processing},
year = {2023},
isbn = {9798400700507},
publisher = {Association for Computing Machinery},
address = {New York, NY, USA},
url = {https://doi.org/10.1145/3576915.3623160},
doi = {10.1145/3576915.3623160},
booktitle = {Proceedings of the 2023 ACM SIGSAC Conference on Computer and Communications Security},
pages = {1437–1451},
numpages = {15},
location = {Copenhagen, Denmark},
series = {CCS '23}
}

@misc{cryptoeprint:2025/1117,
      author = {Suyash Bagad and Quang Dao and Yuval Domb and Justin Thaler},
      title = {Speeding Up Sum-Check Proving},
      howpublished = {Cryptology {ePrint} Archive, Paper 2025/1117},
      year = {2025},
      url = {https://eprint.iacr.org/2025/1117}
}

@inproceedings{virgo,
author = {Zhang, Jiaheng and Liu, Tianyi and Wang, Weijie and Zhang, Yinuo and Song, Dawn and Xie, Xiang and Zhang, Yupeng},
title = {Doubly Efficient Interactive Proofs for General Arithmetic Circuits with Linear Prover Time},
year = {2021},
isbn = {9781450384544},
publisher = {Association for Computing Machinery},
address = {New York, NY, USA},
url = {https://doi.org/10.1145/3460120.3484767},
doi = {10.1145/3460120.3484767},
booktitle = {Proceedings of the 2021 ACM SIGSAC Conference on Computer and Communications Security},
pages = {159–177},
numpages = {19},
location = {Virtual Event, Republic of Korea},
series = {CCS '21}
}

@inproceedings{ghodsi2017safetynets,
  title={SafetyNets: Verifiable Execution of Deep Neural Networks on an Untrusted Cloud},
  author={Ghodsi, Zahra and Gu, Tianyu and Garg, Siddharth},
  booktitle={Advances in Neural Information Processing Systems},
  volume={30},
  pages={4672--4681},
  year={2017}
}

@misc{hydra,
      author = {William Zhang and Yu Xia},
      title = {Hydra: Succinct Fully Pipelineable Interactive Arguments of Knowledge},
      howpublished = {Cryptology {ePrint} Archive, Paper 2021/641},
      year = {2021},
      url = {https://eprint.iacr.org/2021/641}
}

@article{esteva2019guide,
  author    = {Esteva, Andre and Robicquet, Alexandre and Ramsundar, Bharath and Kuleshov, Volodymyr and DePristo, Mark and Chou, Katherine and Cui, Claire and Corrado, Greg and Thrun, Sebastian and Dean, Jeff},
  title     = {A Guide to Deep Learning in Healthcare},
  journal   = {Nature Medicine},
  volume    = {25},
  number    = {1},
  pages     = {24--29},
  year      = {2019},
  publisher = {Nature Publishing Group}
}

@article{abdallah2016fraud,
  author    = {Abdallah, Aisha and Maarof, Mohd Aizaini and Zainal, Anazida},
  title     = {Fraud Detection System: A Survey},
  journal   = {Journal of Network and Computer Applications},
  volume    = {68},
  pages     = {90--113},
  year      = {2016},
  publisher = {Elsevier}
}

@article{kuo2019deep,
  author    = {Kuo, Kevin},
  title     = {DeepTriangle: A Deep Learning Approach to Loss Reserving},
  journal   = {Risks},
  volume    = {7},
  number    = {3},
  pages     = {97},
  year      = {2019},
  publisher = {MDPI}
}

@misc{feng2021zen,
  author       = {Boyuan Feng and Lianke Qin and Zhenfei Zhang and Yufei Ding and Shumo Chu},
  title        = {ZEN: Efficient Zero-Knowledge Proofs for Neural Networks},
  howpublished = {Cryptology ePrint Archive, Paper 2021/087},
  year         = {2021},
  url          = {https://eprint.iacr.org/2021/087}
}

@inproceedings{feng2024zeno,
  author    = {Boyuan Feng and Zheng Wang and Yuke Wang and Shu Yang and Yufei Ding},
  title     = {ZENO: A Type-based Optimization Framework for Zero-Knowledge Neural Network Inference},
  booktitle = {Proceedings of the 29th ACM International Conference on Architectural Support for Programming Languages and Operating Systems (ASPLOS '24)},
  year      = {2024},
  pages     = {450--464},
  publisher = {ACM},
  doi       = {10.1145/3617232.3624852}
}

@misc{li2025zkpog,
      author = {Muyang Li and Yueteng Yu and Bangyan Wang and Xiong Fan and Shuwen Deng},
      title = {{ZKPoG}: Accelerating {WitGen}-Incorporated End-to-End Zero-Knowledge Proof on {GPU}},
      howpublished = {Cryptology {ePrint} Archive, Paper 2025/765},
      year = {2025},
      url = {https://eprint.iacr.org/2025/765}
}

@inproceedings{samardzic2024nocap,
  author={Samardzic, Nikola and Langowski, Simon and Devadas, Srinivas and Sanchez, Daniel},
  booktitle={2024 57th IEEE/ACM International Symposium on Microarchitecture (MICRO)}, 
  title={Accelerating Zero-Knowledge Proofs Through Hardware-Algorithm Co-Design}, 
  year={2024},
  volume={},
  number={},
  pages={366-379},
  doi={10.1109/MICRO61859.2024.00035}}

@inproceedings{daftardar2024szkp,
  author    = {Alhad Daftardar and Shikhar Kumar and Yilong Li and Brandon Reagen and Siddharth Garg},
  title     = {SZKP: A Scalable Accelerator Architecture for Zero-Knowledge Proofs},
  booktitle = {Proceedings of the 33rd International Conference on Parallel Architectures and Compilation Techniques (PACT '24)},
  year      = {2024},
  pages     = {271--283},
  publisher = {ACM},
  doi       = {10.1145/3656019.3676898}
}

@inproceedings{wang2025unizk,
  author    = {Cheng Wang and Mingyu Gao},
  title     = {UniZK: Accelerating Zero-Knowledge Proof with Unified Hardware and Flexible Kernel Mapping},
  booktitle = {Proceedings of the 30th ACM International Conference on Architectural Support for Programming Languages and Operating Systems (ASPLOS '25)},
  year      = {2025},
  publisher = {ACM},
  doi       = {10.1145/3669940.3707228}
}

@inproceedings{daftardar2025zkspeed,
  author    = {Alhad Daftardar and Jianqiao Mo and Joey Ah-kiow and Benedikt B{\"u}nz and Ramesh Karri and Siddharth Garg and Brandon Reagen},
  title     = {Need for zkSpeed: Accelerating HyperPlonk for Zero-Knowledge Proofs},
  booktitle = {Proceedings of the 52nd Annual International Symposium on Computer Architecture (ISCA '25)},
  year      = {2025},
  pages     = {1986--2001},
  publisher = {ACM},
  doi       = {10.1145/3695053.3731021}
}

@InProceedings{bn254,
author="Barreto, Paulo S. L. M.
and Naehrig, Michael",
editor="Preneel, Bart
and Tavares, Stafford",
title="Pairing-Friendly Elliptic Curves of Prime Order",
booktitle="Selected Areas in Cryptography",
year="2006",
publisher="Springer Berlin Heidelberg",
address="Berlin, Heidelberg",
pages="319--331",
isbn="978-3-540-33109-4"
}

@misc{zksumcheck,
  author = {Alessandro Chiesa and Michael A. Forbes and Nicholas Spooner},
  title = {A Zero Knowledge Sumcheck and its Applications},
  howpublished = {Cryptology ePrint Archive, Paper 2017/305},
  year = {2017},
  url = {https://eprint.iacr.org/2017/305}
}

@InProceedings{fs_attack,
author="Khovratovich, Dmitry
and Rothblum, Ron D.
and Soukhanov, Lev",
editor="Tauman Kalai, Yael
and Kamara, Seny F.",
title="How to Prove False Statements: Practical Attacks on Fiat-Shamir",
booktitle="Advances in Cryptology -- CRYPTO 2025",
year="2025",
publisher="Springer Nature Switzerland",
address="Cham",
pages="3--26",
isbn="978-3-032-01887-8"
}

@inproceedings{fiat1986prove,
  title={How to prove yourself: Practical solutions to identification and signature problems},
  author={Fiat, Amos and Shamir, Adi},
  booktitle={Conference on the Theory and Application of Cryptographic Techniques},
  pages={186--194},
  year={1986},
  organization={Springer}
}

@article{Thaler2022,
  author = {Justin Thaler},
  title = {Proofs, Arguments, and Zero-Knowledge},
  journal = {Foundations and Trends® in Privacy and Security},
  volume = {4},
  number = {2--4},
  pages = {117--660},
  year = {2022},
  publisher = {Now Publishers},
  doi = {10.1561/3300000030},
  url = {http://dx.doi.org/10.1561/3300000030}
}

@article{GKR15,
  author = {Shafi Goldwasser and Yael Tauman Kalai and Guy N. Rothblum},
  title = {Delegating Computation: Interactive Proofs for Muggles},
  journal = {Journal of the ACM},
  volume = {62},
  number = {4},
  pages = {27:1--27:64},
  year = {2015},
  publisher = {ACM},
  doi = {10.1145/2699436}
}

@inproceedings{Groth16,
  author = {Jens Groth},
  title = {On the Size of Pairing-Based Non-interactive Arguments},
  booktitle = {Advances in Cryptology -- EUROCRYPT 2016},
  series = {Lecture Notes in Computer Science},
  volume = {9666},
  pages = {305--326},
  publisher = {Springer},
  year = {2016},
  doi = {10.1007/978-3-662-49896-5_11}
}

@article{sumcheck,
  author = {Carsten Lund and Lance Fortnow and Howard Karloff and Noam Nisan},
  title = {Algebraic Methods for Interactive Proof Systems},
  journal = {Journal of the ACM},
  volume = {39},
  number = {4},
  pages = {859--868},
  year = {1992},
  month = {October},
  doi = {10.1145/146585.146605}
}

@inproceedings{spartan,
  title={Spartan: Efficient and General-Purpose zkSNARKs Without Trusted Setup},
  author={Srinath T. V. Setty},
  booktitle={Advances in Cryptology -- CRYPTO 2020},
  series={Lecture Notes in Computer Science},
  volume={12172},
  pages={704--737},
  year={2020},
  publisher={Springer},
  doi={10.1007/978-3-030-56877-1_25}
}

@misc{zktorch,
      title={ZKTorch: Compiling ML Inference to Zero-Knowledge Proofs via Parallel Proof Accumulation}, 
      author={Bing-Jyue Chen and Lilia Tang and Daniel Kang},
      year={2025},
      eprint={2507.07031},
      archivePrefix={arXiv},
      primaryClass={cs.CR},
      url={https://arxiv.org/abs/2507.07031}, 
}

@misc{cai2025gettingpayforauditing,
      title={Are You Getting What You Pay For? Auditing Model Substitution in LLM APIs}, 
      author={Will Cai and Tianneng Shi and Xuandong Zhao and Dawn Song},
      year={2025},
      eprint={2504.04715},
      archivePrefix={arXiv},
      primaryClass={cs.CL},
      url={https://arxiv.org/abs/2504.04715}, 
}

@misc{gpt_cost,
      title={The rising costs of training frontier AI models}, 
      author={Ben Cottier and Robi Rahman and Loredana Fattorini and Nestor Maslej and Tamay Besiroglu and David Owen},
      year={2025},
      eprint={2405.21015},
      archivePrefix={arXiv},
      primaryClass={cs.CY},
      url={https://arxiv.org/abs/2405.21015}, 
}

@article{MATOS2026101001,
title = {Critical appraisal of fairness metrics for artificial intelligence-based clinical prediction models: a scoping review},
journal = {The Lancet Digital Health},
pages = {101001},
year = {2026},
issn = {2589-7500},
doi = {https://doi.org/10.1016/j.landig.2026.101001},
url = {https://www.sciencedirect.com/science/article/pii/S2589750026000245},
author = {João Matos and Ben {Van Calster} and Leo Anthony Celi and Paula Dhiman and Judy Wawira Gichoya and Richard D Riley and Chris Russell and Sara Khalid and Gary S Collins}
}

@inproceedings{libra,
  author    = {Tiancheng Xie and
               Jiaheng Zhang and
               Yupeng Zhang and
               Charalampos Papamanthou and
               Dawn Song},
  title     = {Libra: Succinct Zero-Knowledge Proofs with Optimal Prover Computation},
  booktitle = {Advances in Cryptology -- {CRYPTO} 2019},
  series    = {Lecture Notes in Computer Science},
  volume    = {11694},
  pages     = {733--764},
  publisher = {Springer},
  year      = {2019},
  doi       = {10.1007/978-3-030-26954-8_24}
}

@inproceedings{zkgpt,
  author = {Wenjie Qu and Yijun Sun and Xuanming Liu and Tao Lu and Yanpei Guo and Kai Chen and Jiaheng Zhang},
  title = {{zkGPT}: An Efficient Non-interactive Zero-knowledge Proof Framework for {LLM} Inference},
  booktitle = {34th USENIX Security Symposium (USENIX Security 25)},
  year = {2025},
  url = {https://www.usenix.org/conference/usenixsecurity25/presentation/qu-zkgpt}
}

@inproceedings{hyrax,
  title={Doubly-Efficient zkSNARKs Without Trusted Setup},
  author={Riad S. Wahby and Ioanna Tzialla and abhi shelat and Justin Thaler and Michael Walfish},
  booktitle={2018 IEEE Symposium on Security and Privacy (SP)},
  pages={926--943},
  year={2018},
  doi={10.1109/SP.2018.00060}
}

@inproceedings{campanelli2019legosnark,
  author    = {Matteo Campanelli and Dario Fiore and Anais Querol},
  title     = {LegoSNARK: Modular Design and Composition of Succinct Zero-Knowledge Proofs},
  booktitle = {Proceedings of the 2019 ACM SIGSAC Conference on Computer and Communications Security (CCS)},
  year      = {2019},
  pages     = {2075--2092}
}

@inproceedings{lasso,
    author    = {Setty, Srinath and Thaler, Justin and Wahby, Riad},
    title     = {Unlocking the Lookup Singularity with Lasso},
    booktitle = {Advances in Cryptology -- EUROCRYPT 2024},
    year      = {2024},
    publisher = {Springer},
    address   = {Cham}
}

@inproceedings{ma2023gzkp,
author = {Ma, Weiliang and Xiong, Qian and Shi, Xuanhua and Ma, Xiaosong and Jin, Hai and Kuang, Haozhao and Gao, Mingyu and Zhang, Ye and Shen, Haichen and Hu, Weifang},
title = {GZKP: A GPU Accelerated Zero-Knowledge Proof System},
year = {2023},
isbn = {9781450399166},
publisher = {Association for Computing Machinery},
address = {New York, NY, USA},
url = {https://doi.org/10.1145/3575693.3575711},
doi = {10.1145/3575693.3575711},
booktitle = {Proceedings of the 28th ACM International Conference on Architectural Support for Programming Languages and Operating Systems, Volume 2},
pages = {340–353},
numpages = {14},
location = {Vancouver, BC, Canada},
series = {ASPLOS 2023}
}

@misc{kthiri2024icicle,
  author       = {Kthiri, Mohamed},
  title        = {Benchmarking GPU Acceleration for ZK-SNARKs with Icicle},
  howpublished = {\url{https://www.maya-zk.com/blog/gpu-acceleration}},
  year         = {2024},
  note         = {Accessed: 2025-08-30}
}

@misc{apollo,
      title={Artemis: Efficient Commit-and-Prove SNARKs for zkML}, 
      author={Hidde Lycklama and Alexander Viand and Nikolay Avramov and Nicolas Küchler and Anwar Hithnawi},
      year={2025},
      eprint={2409.12055},
      archivePrefix={arXiv},
      primaryClass={cs.CR},
      url={https://arxiv.org/abs/2409.12055}, 
}

@inproceedings{zkcnn,
  author = {Liu, Tianyi and Xie, Xiang and Zhang, Yupeng},
  title = {{zkCNN}: Zero Knowledge Proofs for Convolutional Neural Network Predictions and Accuracy},
  booktitle = {Proceedings of the 2021 ACM SIGSAC Conference on Computer and Communications Security},
  series = {CCS '21},
  year = {2021},
  pages = {2968--2985},
  doi = {10.1145/3460120.3485379},
  url = {https://eprint.iacr.org/2021/673}
}

@inproceedings{zkllm,
author = {Sun, Haochen and Li, Jason and Zhang, Hongyang},
title = {zkLLM: Zero Knowledge Proofs for Large Language Models},
year = {2024},
isbn = {9798400706363},
publisher = {Association for Computing Machinery},
address = {New York, NY, USA},
url = {https://doi.org/10.1145/3658644.3670334},
doi = {10.1145/3658644.3670334},
booktitle = {Proceedings of the 2024 on ACM SIGSAC Conference on Computer and Communications Security},
pages = {4405–4419},
numpages = {15},
location = {Salt Lake City, UT, USA},
series = {CCS '24}
}

@inproceedings{zkml,
  author = {Kang, Daniel and Hashimoto, Tatsunori and Stoica, Ion and Sun, Yi},
  title = {{ZKML}: An Optimizing System for {ML} Inference in Zero-Knowledge Proofs},
  booktitle = {Proceedings of the Nineteenth European Conference on Computer Systems},
  series = {EuroSys '24},
  year = {2024},
  month = apr,
  location = {Athens, Greece},
  pages = {622--636},
  doi = {10.1145/3627703.3650088}
}

@article{goldwasser1989knowledge,
  title={The knowledge complexity of interactive proof systems},
  author={Goldwasser, Shafi and Micali, Silvio and Rackoff, Charles},
  journal={SIAM Journal on Computing},
  volume={18},
  number={1},
  pages={186--208},
  year={1989},
  publisher={SIAM}
}

@techreport{sha3,
	title = {{SHA}-3 standard : permutation-based hash and extendable-output functions},
	shorttitle = {{SHA}-3 standard},
	url = {https://nvlpubs.nist.gov/nistpubs/FIPS/NIST.FIPS.202.pdf},
	urldate = {2026-01-24},
	institution = {National Institute of Standards and Technology (U.S.)},
	author = {{National Institute of Standards and Technology (US)}},
	year = {2015},
	doi = {10.6028/NIST.FIPS.202},
}

@misc{bls12381,
  author = {Bowe, Sean},
  title = {{BLS12-381}: New zk-{SNARK} Elliptic Curve Construction},
  howpublished = {Zcash Blog},
  year = {2017},
  month = mar,
  url = {https://electriccoin.co/blog/new-snark-curve/}
}

@inproceedings{campanelli2021lunar,
author = {Campanelli, Matteo and Faonio, Antonio and Fiore, Dario and Querol, Ana\"{\i}s and Rodr\'{\i}guez, Hadri\'{a}n},
title = {Lunar: A Toolbox for More Efficient Universal and Updatable zkSNARKs and Commit-and-Prove Extensions},
year = {2021},
isbn = {978-3-030-92077-7},
publisher = {Springer-Verlag},
address = {Berlin, Heidelberg},
url = {https://doi.org/10.1007/978-3-030-92078-4_1},
doi = {10.1007/978-3-030-92078-4_1},
booktitle = {Advances in Cryptology – ASIACRYPT 2021: 27th International Conference on the Theory and Application of Cryptology and Information Security, Singapore, December 6–10, 2021, Proceedings, Part III},
pages = {3–33},
numpages = {31},
location = {Singapore, Singapore}
}

@article{Gabizon2019PLONKPO,
  title={PLONK: Permutations over Lagrange-bases for Oecumenical Noninteractive arguments of Knowledge},
  author={Ariel Gabizon and Zachary J. Williamson and Oana-Madalina Ciobotaru},
  journal={IACR Cryptol. ePrint Arch.},
  year={2019},
  volume={2019},
  pages={953},
  url={https://api.semanticscholar.org/CorpusID:201685538}
}

@misc{mcl,
	title = {herumi/mcl},
	copyright = {BSD-3-Clause},
	url = {https://github.com/herumi/mcl},
	urldate = {2026-01-24},
	author = {Shigeo, MITSUNARI},
	month = dec,
	year = {2025},
	note = {original-date: 2015-05-05T00:18:39Z},
}

@online{gkr_chunking,
  author       = {{Polyhedra Network}},
  title        = {{GKR} Input Layer Chunks},
  year         = {n.d.},
  url          = {https://docs.polyhedra.network/expander/prover_internals/input_chunks/},
  urldate      = {2026-01-23},
  organization = {Polyhedra Network Documentation}
}

@INPROCEEDINGS{resnet,
  author={He, Kaiming and Zhang, Xiangyu and Ren, Shaoqing and Sun, Jian},
  booktitle={2016 IEEE Conference on Computer Vision and Pattern Recognition (CVPR)}, 
  title={Deep Residual Learning for Image Recognition}, 
  year={2016},
  volume={},
  number={},
  pages={770-778},
  doi={10.1109/CVPR.2016.90}}
